\documentclass{article}

\usepackage{times}
\usepackage{amssymb,amsmath,amsthm}
\usepackage{mathrsfs}
\usepackage{epsfig}
\usepackage{color}

\newcommand\ie{{\em i.e.}~}
\newcommand\eg{{\em e.g.}~}
\newcommand\cf{{\em cf.}~}

\def\B{\mathscr B}

\def\C{\mathbb C}
\def\CC{\mathscr C}
\def\d{\mathrm{d}}
\def\D{\mathscr D}

\def\E{\mathcal E}

\def\F{\mathscr F}

\def\G{\mathcal G}
\def\H{\mathcal H}

\def\j{\mathsf j}
\def\K{\mathscr K}
\def\L{\mathscr L}

\def\N{\mathbb N}
\def\O{\mathcal O}
\def\P{\mathcal P}

\def\R{\mathbb R}
\def\S{\mathscr S}
\def\T{\mathcal T}

\def\U{\mathcal U}
\def\V{\mathcal V}
\def\X{\mathfrak X}

\def\12{{\textstyle\frac12}}
\def\dom{\mathcal D}
\def\lone{\mathsf{L}^{\:\!\!1}}

\def\ltwo{\mathsf{L}^{\:\!\!2}}
\def\ltwoloc{\mathsf{L}^{\:\!\!2}_{\rm loc}}
\def\linf{\mathsf{L}^{\:\!\!\infty}}

\def\e{\mathop{\mathrm{e}}\nolimits}

\def\supp{\mathop{\mathrm{supp}}\nolimits}

\def\slim{\mathop{\hbox{\rm s-}\lim}\nolimits}

\def\ad{\mathop{\mathrm{ad}}\nolimits}
\renewcommand\Delta{\triangle}

\def\ds{\mathrm{ds}}
\def\Mc{M_{\rm c}}
\def\Minf{{M_\infty}}
\def\dv{\mathrm{dv}}
\def\id{\mathrm{id}_\mathbb{R}}
\def\div{\mathrm{div}}
\def\Ran{\mathop{\mathsf{Ran}}\nolimits}

\def\ad{{\rm ad}}
\def\muL{\mu_{\rm L}}
\def\muS{\mu_{\rm S}}

\newtheorem{Theorem}{Theorem}[section]
\newtheorem{Remark}[Theorem]{Remark}
\newtheorem{Lemma}[Theorem]{Lemma}
\newtheorem{Assumption}[Theorem]{Assumption}
\newtheorem{Corollary}[Theorem]{Corollary}
\newtheorem{Proposition}[Theorem]{Proposition}

\begin{document}


\title{Spectral analysis and time-dependent scattering theory\\
on manifolds with asymptotically cylindrical ends}

\author{S. Richard$^1$\footnote{
Supported by the Japan Society for the Promotion of Science and by Grants-in-Aid for
scientific Research. A two weeks stay in Santiago de Chile and a one week stay at the
Centre Interfacultaire Bernoulli (EPFL, Lausanne) are also acknowledged.}~~and
R. Tiedra de Aldecoa$^2$\footnote{Supported by the Fondecyt
Grant 1090008 and by the Iniciativa Cientifica Milenio ICM P07-027-F ``Mathematical
Theory of Quantum and Classical Magnetic Systems''.}}

\date{\small}
\maketitle \vspace{-1cm}

\begin{quote}
\emph{
\begin{itemize}
\item[$^1$] Graduate School of Pure and Applied Sciences,
University of Tsukuba,
1-1-1 Tennodai, Tsukuba,
Ibaraki 305-8571, Japan. \\
On leave from Universit\'e de Lyon; Universit\'e
Lyon 1; CNRS, UMR5208, Institut Camille Jordan, 43 blvd du 11 novembre 1918, F-69622
Villeurbanne-Cedex, France
\item[$^2$] Facultad de Matem\'aticas, Pontificia Universidad Cat\'olica de Chile,\\
Av. Vicu\~na Mackenna 4860, Santiago, Chile
\item[] \emph{E-mails:} richard@math.univ-lyon1.fr, rtiedra@mat.puc.cl
\end{itemize}
}
\end{quote}


\begin{abstract}
We review the spectral analysis and the time-dependent approach of scattering theory
for manifolds with asymptotically cylindrical ends. For the spectral analysis, higher
order resolvent estimates are obtained via Mourre theory for both short-range and
long-range behaviors of the metric and the perturbation at infinity. For the
scattering theory, the existence and asymptotic completeness of the wave operators is
proved in a two-Hilbert spaces setting. A stationary formula as well as mapping
properties for the scattering operator are derived. The existence of time delay and
its equality with the Eisenbud-Wigner time delay is finally presented. Our analysis
mainly differs from the existing literature on the choice of a simpler comparison
dynamics as well as on the complementary use of time-dependent and stationary
scattering theories.
\end{abstract}

\textbf{2000 Mathematics Subject Classification:} 58J50, 81Q10, 47A40.

\smallskip

\textbf{Keywords:} Manifolds, spectral analysis, scattering theory, conjugate
operator.


\section{Introduction}\label{Intro}
\setcounter{equation}{0}

Manifolds with asymptotically cylindrical ends are certainly some of the most studied
manifolds in spectral and scattering theory, and many results related to them are
already available in the literature, see for example
\cite{Ch95,Ch02, Ch09,CZ95,Gui89,IKL10,Mel95,Mul98,MS10}. The aim of the present paper
is to complement this bulk of information and to apply recent technics or results in
commutator methods, time-dependent scattering theory, stationary methods and quantum
time delay to these manifolds. As examples of new results, we provide higher order
resolvent estimates for both short-range and long-range behaviors of the metric and the
perturbation at infinity, we deduce mapping properties of the scattering operator,
and we also prove the existence and the equality of global and Eisenbud-Wigner time
delays. Also, we emphasize that our analysis differs from much of the existing
literature on the choice of a simpler reference dynamics.

At the origin of this research stand our three recent works on spectral and scattering
theory in an abstract framework \cite{RT10,RT11,RT12}. In the first two article, it
is shown that, given a scattering process, particular choices of asymptotic reference
systems are better suited than others and automatically lead to richer results. On
manifolds with asymptotically cylindrical ends, this idea can be particularly well
illustrated. In the third article, a comparison scheme for deducing a Mourre estimate
for a pair of self-adjoint operators $(H,A)$ in a Hilbert space $\H$ from a similar
estimate for a second pair of operators $(H_0,A_0)$ in an auxiliary Hilbert space
$\H_0$ has been put into evidence. Again, a clever choice of the reference system
$(\H_0,H_0)$ is of much help. However, this comparison scheme, though at the root of
the time-dependent scattering theory, has not yet been systematically implemented in
Mourre theory. This article can also be regarded as an attempt to fill in this gap in
the context of manifolds with asymptotically cylindrical ends (see also
\cite{DHS92,Doi99,Don99,FH89,IN10} for related works).

Let us now be more precise about the model. We consider a smooth, non-compact,
complete Riemannian manifold $M$ of dimension $n+1\ge2$ without boundary. We assume
that $M$ is of the form $M=\Mc\cup \Minf$, with $\Mc$ relatively compact and $\Minf$
open in $M$. Moreover, we suppose $\Minf$ diffeomorphic to $(0,\infty)\times\Sigma$,
with $\Sigma$ the disjoint union of a finite number of smooth, compact, connected
Riemannian manifolds of dimension $n\ge1$ without boundary. The Riemannian metric
$g|_{M_\infty}$ on $M_\infty$ converges at infinity (in a suitable sense) to the
product metric on $(0,\infty)\times\Sigma$. The usual volume form on $M$ is denoted
by $\dv$, while the one on $\Sigma$ is denoted by $\ds$. In the Hilbert space
$\H:=\ltwo(M,\dv)$, we consider the self-adjoint operator $H:=\Delta_M+V$, where
$\Delta_M$ is the (Dirichlet) Laplace-Beltrami operator on $M$ and $V$ is a
multiplication operator by a smooth bounded function on $M$.

As a reference system, we consider the Laplace-Beltrami operator
$H_0:=\Delta_{\R\times\Sigma}$ in the Hilbert space
$\H_0:=\ltwo(\R\times\Sigma,\d x\otimes\ds)$. This choice of reference system instead
of the more usual Laplacian $\Delta_{(0,\infty)\times\Sigma}$ in
$\ltwo\big((0,\infty)\times\Sigma,\d x\otimes\ds\big)$ with a Neumann or Dirichlet
condition at the origin is inspired by the following considerations. On the first
hand, it involves no arbitrariness when defining $H_0$, since the Laplacian
$\Delta_{\R\times\Sigma}$ is the only natural choice for the comparison operator in
$\ltwo(\R\times\Sigma,\d x\otimes\ds)$. On the second hand, it allows to take
constantly advantage of the existence of a simple conjugate operator $A_0$ for $H_0$
and a simple spectral representation for $H_0$. Finally, it permits to define easily
a family $\{H_0(x)\}_{x\in\R}$ of mutually commuting self-adjoint operators in
$\H_0$, which plays an important role for the proof of the existence of quantum time
delay (the operators $H(x)$ are simply the translated operators
$\e^{-ix\Phi_0}H_0\e^{ix\Phi_0}$, with $\Phi_0:=Q\otimes1$ and $Q$ the position
operator in $\ltwo(\R,\d x)$).

In order to link the dynamics $H$ in $\H$ to the reference dynamics $H_0$ in $\H_0$
we use, as is usual in scattering theory, an identification operator
$J\in\B(\H_0,\H)$. Essentially, $J$ acts as the zero operator on vectors
$\varphi\in\H_0$ having support in $(-\infty,1)\times\Sigma$ and maps isometrically
vectors $\varphi\in\H_0$ having support in $(2,\infty)\times\Sigma$ onto vectors
$J\varphi\in\H$ having support in $M_\infty$. With these tools at hand, and by using
extensively the two-Hilbert scheme of \cite{RT12}, we are able to establish various
novel results for the operator $H$ and the scattering triple $(H_0,H,J)$ that we now
describe.

In Section \ref{sec_spect}, we perform the spectral analysis of $H$ when both the
metric $g$ and the potential $V$ are the sum of two terms, one having a short-range
type behavior at infinity and one having a long-range type behaviour at infinity. We
start in Section \ref{Sec_conjugate} by defining an appropriate conjugate operator
$A$ for $H$. Following the general scheme of \cite[Sec.~3]{RT12}, we simply use the
operator $A=JA_0J^*$, with $A_0$ the generator of dilations along the $\R$-axis in
$\H_0$. With this operator, we establish a Mourre estimate for $H$ in Proposition
\ref{Mourre}. Then, by using an abstract result of \cite{BG96}, we prove the
Zygmund-H\"older regularity of the map
\begin{equation}\label{ZygZyg}
\R\ni\lambda\mapsto
\langle A\rangle^{-s}(H-\lambda\mp i0)^{-1}\langle A\rangle^{-s}\in\B(\H)
\end{equation}
for suitable $s$ and away from the critical values of $H$. This result implies in
particular higher order resolvent estimates for $H$ and higher order
differentiability of the map \eqref{ZygZyg} (see Proposition \ref{yeye} for a precise
statement). As a by-product, formulated in Proposition \ref{natureH}, we obtain the
absence of singular continuous spectrum and the finiteness of the point spectrum of
$H$ away from the set $\T$ of eigenvalues of the (transverse) Laplacian
$\Delta_\Sigma$ on $\Sigma$. In the particular case where the metric $g|_{M_\infty}$
is purely short-range with decay $\langle x\rangle^{-\mu}$, $\mu>1$, at infinity this
result is comparable with the one recently obtained in \cite[Thm.~3.10]{IKL10} with
alternative technics.

In Section \ref{Sec_Scatt}, we present the time-dependent scattering theory for the
triple $(H_0,H,J)$ when the metric $g|_{M_\infty}$ on $M_\infty$ decays as
$\langle x\rangle^{-\mu}$, $\mu>1$, at infinity. In Proposition \ref{existence}, we
prove that the generalized wave operators
$$
W_\pm:=\slim_{t\to\pm\infty}\e^{itH}J\e^{-itH_0}
$$
exist and are partial isometries with initial subspaces
$
\H_0^\pm
:=\big\{\varphi\in\H_0\mid\supp(\F\otimes1)\varphi\subset\R_\pm\times\Sigma\big\}
$. Here $\F$ denotes the Fourier transform in $\ltwo(\R)$. Then, we establish in
Proposition \ref{asymptotc} the asymptotic completeness of the wave operators $W_\pm$
by using an abstract criterion of \cite{RT12}. This implies in particular the
existence and the unitarity of the scattering operator $S:=W_+^*W_-:\H_0^-\to\H_0^+$.
In Section \ref{Sec_Sta-Sta}, we pursue our study by deriving a precise stationary
formula for the scattering matrix $S(\lambda)$ at energy $\lambda$ (see Theorem
\ref{formule_S}). This formula, together with the Zygmund-H\"older regularity of the
resolvent map, allows us to prove that the map $\lambda\mapsto S(\lambda)$ is locally
$k$-times H\"older continuously differentiable away from the critical values of $H$
if $\mu>k+1$ (see Corollary \ref{Cor_diff} for details). This result implies in turns
a mapping property of the scattering operator $S$, which is crucial (and usually
considered as the difficult part) for the proof of the existence of time delay (see
Proposition \ref{Prop_map}). Finally, we prove in Section \ref{Sec_Time} the
existence of time delay and its equality with Eisenbud-Wigner time delay using the
abstract method of \cite{RT11}.

As a final comment, let us stress that even if manifolds with asymptotically
cylindrical ends are certainly a piece of folklore for experts in global analysis,
most of the results contained in this paper are either new or presented in a more
systematic form than the ones already existing in the literature. Furthermore, the
abstract framework underlying our analysis as well as the our scheme of
investigations can serve again for further investigations on other types of
manifolds. We intend to perform such investigations in the near future.\\

\noindent
{\bf Notations:}
$\S(\R)$ denotes the Schwartz space on $\R$. The operators $P$ and $Q$ are
respectively the momentum and the position operators in $\ltwo(\R)$, \ie
$(P\varphi)(x):=-i\varphi'(x)$ and $(Q\varphi)(x):=x\varphi(x)$ for each
$\varphi\in\S(\R)$ and $x\in\R$. $\N:=\{0,1,2,\ldots\}$ is the set of natural numbers
and $\H^s_t(\R)$, $s,t\in\R$, are the weighted Sobolev spaces over $\R$
\cite[Sec.~4.1]{ABG} (with the convention that $\H^s(\R):=\H^s_0(\R)$ and
$\H_t(\R):=\H^0_t(\R)$). The one-dimensional Fourier transform $\F$ is a topological
isomorphism of $\H^s_t(\R)$ onto $\H^t_s(\R)$ for any $s,t\in\R$. Finally, $\otimes$
(resp. $\odot$) stands for the closed (resp. algebraic) tensor product of Hilbert
spaces or of operators.

\section{Reference system}\label{SecFree}
\setcounter{equation}{0}

We introduce in this section the asymptotic reference system $(\H_0,H_0)$. As
explained in the introduction, the configuration space subjacent to the Hilbert
space $\H_0$ is a direct product $\R\times\Sigma$, where $\Sigma$ is the disjoint
union of $N\ge1$ Riemannian manifolds $\Sigma_\ell$. So, we start by defining each
manifold $\Sigma_\ell$ separately.

Let $(\Sigma_\ell,h_\ell)$ be a smooth, compact, orientable, connected Riemannian
manifold of dimension $n\ge1$, without boundary. On a chart $(\O_\ell,\omega_\ell)$
of $\Sigma_\ell$, the Riemannian metric
$h_\ell:\X(\Sigma_\ell)\otimes \X(\Sigma_\ell)\to C^\infty(\Sigma_\ell)$ is given by
the collection of functions $(h_\ell)_{jk}\in C^\infty(\O_\ell)$,
$j,k\in\{1,\ldots,n\}$, defined by
$$
(h_\ell)_{jk}:=h_\ell\left(\frac\partial{\partial\omega_\ell^j},
\frac\partial{\partial\omega_\ell^k}\right).
$$
The contravariant form of the metric tensor $h_\ell$ has components $(h_\ell)^{jk}$
determined by the matrix relation $\sum_i(h_\ell)_{ij}(h_\ell)^{ik}=\delta^k_j$, and
the volume element $\ds_\ell$ on $\Sigma_\ell$ is given by
$$
\ds_\ell:={\frak h}_\ell\,\d\omega_\ell\quad\hbox{with}\quad
{\frak h}_\ell:=\sqrt{\det\big\{(h_\ell)_{jk}\big\}}.
$$
The Laplace-Beltrami operator $\Delta_{\Sigma_\ell}$ in the Hilbert space
$\ltwo(\Sigma_\ell):=\ltwo(\Sigma_\ell,\ds_\ell)$ is defined on each chart by
$$
\Delta_{\Sigma_\ell}\varphi:=-\sum_{j,k=1}^n{\frak h}_\ell^{-1}
\frac\partial{\partial\omega_\ell^j}\;\!{\frak h}_\ell(h_\ell)^{jk}
\frac\partial{\partial\omega_\ell^k}\;\!\varphi,
\quad\varphi\in C^\infty(\Sigma_\ell).
$$
It is known that $\Delta_{\Sigma_\ell}$ is essentially self-adjoint on
$C^\infty(\Sigma_\ell)$ \cite[Thm.~3]{Cor72} and that the closure of
$\Delta_{\Sigma_\ell}$ (which we denote by the same symbol) has a spectrum
$\sigma(\Delta_{\Sigma_\ell})$ consisting in an unbounded sequence of finitely
degenerated eigenvalues
$
0=\tau_{\ell,0}<\tau_{\ell,1}\le\tau_{\ell,2}\le\ldots
$
repeated according to multiplicity \cite[Thm.~1.29]{Ros97}.

For $N\ge 1$, let  $\Sigma:=\bigsqcup_{\ell=1}^N\Sigma_\ell$ be the disjoint union
of the manifolds $\Sigma_\ell$. When endowed with the metric $h$ defined by
$$
[h(X,Y)](\ell,p):=(h_\ell)_p\big(X_{(\ell,p)},Y_{(\ell,p)}\big),\quad
(\ell,p)\in\Sigma,~X_{(\ell,p)},Y_{(\ell,p)}\in T_p\Sigma_\ell,
$$
the set $\Sigma$ becomes a Riemannian manifold. Its volume element $\ds$ is given by
$$
\ds(\V_1,\ldots,\V_N):=\sum_{\ell=1}^N\ds_\ell(\V_\ell),
\quad(\V_1,\ldots,\V_N)\subset\Sigma.
$$
The Laplace-Beltrami operator
$\Delta_\Sigma\simeq\bigoplus_{\ell=1}^N\Delta_{\Sigma_\ell}$ in
$
\ltwo(\Sigma)
:=\ltwo(\Sigma,\ds)\simeq\bigoplus_{\ell=1}^N\ltwo(\Sigma_\ell,\ds_\ell)
$
is essentially self-adjoint on
$C^\infty(\Sigma)\simeq\bigoplus_{\ell=1}^NC^\infty(\Sigma_\ell)$ and has purely
discrete spectrum $\T:=\{\tau_j\}_{j\in\N}$ (the values $\tau_j$ being the elements
of $\big\{\tau_{\ell,k}\mid\ell=1,\ldots,N,~k\in\N\big\}$ arranged in ascending order
and repeated according to multiplicity).

Then, we define in the Hilbert space
$
\H_0:=\ltwo(\R\times\Sigma,\d x\otimes\ds)\simeq\ltwo(\R)\otimes\ltwo(\Sigma)
$
the operator $H_0:=P^2\otimes1+1\otimes\Delta_\Sigma$. The operator $H_0$ is
essentially self-adjoint on $\S(\R)\odot C^\infty(\Sigma)$ and has domain
\cite[Sec.~3]{BG92}
$$
\dom(H_0)
=\big\{\ltwo(\R)\otimes\dom(\Delta_\Sigma)\big\}
\cap\big\{\H^2(\R)\otimes\ltwo(\Sigma)\big\},
$$
endowed with the intersection topology. The spectral measure of $H_0$ is purely
absolutely continuous and admits the tensorial decomposition \cite[Ex.~8.21]{Wei80}:
\begin{equation}\label{decompo}
E^{H_0}(\;\!\cdot\;\!)
=\sum_{j\in \N}E^{P^2+\tau_j}(\;\!\cdot\;\!)\otimes\P_j,
\end{equation}
where $\{\P_j\}_{j\in\N}$ is the family of one-dimensional eigenprojections of
$\Delta_\Sigma$. In particular, the spectrum $\sigma(H_0)$ and the absolutely
continuous spectrum $\sigma_{\rm ac}(H_0)$ of $H_0$ satisfy the identities:
$$
\sigma(H_0)=\sigma_{\rm ac}(H_0)=[0,\infty).
$$

In order to give some results on the spectral representation of $H_0$, one needs to
introduce extra quantities: The fibre $\H_0(\lambda)$ at energy $\lambda\ge0$ in the
spectral representation of $H_0$ is
$$
\H_0(\lambda):=\bigoplus_{j\in\N(\lambda)}
\big\{\P_j\;\!\ltwo(\Sigma)\oplus\P_j\;\!\ltwo(\Sigma)\big\}
\quad\hbox{with}\quad\N(\lambda):=\big\{j\in\N\mid\tau_j\le\lambda\big\}.
$$
Since $\H_0(\lambda)$ is naturally embedded in
$$
\H_0(\infty):=\bigoplus_{j\in\N}
\big\{\P_j\;\!\ltwo(\Sigma)\oplus\P_j\;\!\ltwo(\Sigma)\big\},
$$
we shall sometimes write $\H_0(\infty)$ instead of $\H_0(\lambda)$.
For $\xi\in\R$, we
let $\gamma(\xi):\S(\R)\to\C$ be the restriction operator given by
$\gamma(\xi)\varphi:=\varphi(\xi)$. For $\lambda\in[0,\infty)\setminus\T$, we define
the operator $T_0(\lambda):\S(\R)\odot\ltwo(\Sigma)\to\H_0(\lambda)$ by
\begin{equation}\label{eq_T_0}
\big[T_0(\lambda)\varphi\big]_j:=(\lambda-\tau_j)^{-1/4}
\big\{\big[\gamma\big(-\sqrt{\lambda-\tau_j}\big)\otimes\P_j\big]\varphi,
\big[\gamma\big(\sqrt{\lambda-\tau_j}\big)\otimes\P_j\big]\varphi\big\},
\quad j\in\N(\lambda).
\end{equation}
We can now state the main properties for the operators
$F_0(\lambda):=2^{-1/2}\;\!T_0(\lambda)(\F\otimes 1)$. For shortness, we write
$\widehat\H_0$ for the Hilbert space
$\int_{[0,\infty)}^\oplus\d\lambda\,\H_0(\lambda)$.

\begin{Lemma}[Spectral transformation for $H_0$]\label{Lem_F_0}
Let $t\in\R$. Then
\begin{enumerate}
\item[(a)] For each $\lambda\in[0,\infty)\setminus\T$ and $s>1/2$, the operator
$F_0(\lambda)$ extends to an element of
$\B\big(\H^t_s(\R)\otimes\ltwo(\Sigma),\H_0(\infty)\big)$.
\item[(b)] For each $s>k+1/2$ with $k\in\N$, the function
$
[0,\infty)\setminus\T\ni\lambda\mapsto F_0(\lambda)
\in\B\big(\H^t_s(\R)\otimes\ltwo(\Sigma),\H_0(\infty)\big)
$
is locally $k$-times H\"older continuously differentiable.
\item[(c)] The mapping $F_0:\H_0\to\widehat \H_0$
given for all $\varphi\in\S(\R)\odot\ltwo(\Sigma)$ and every
$\lambda\in[0,\infty)\setminus\T$ by
$$
(F_0\varphi)(\lambda):=F_0(\lambda)\varphi,
$$
extends to a unitary operator, and
$$
F_0H_0F_0^{-1}=\int_{[0,\infty)}^\oplus\d\lambda\,\lambda\;\!.
$$
Furthermore, for any $\phi\in\widehat \H_0$ with
$\phi(\lambda)=\big\{\phi(\lambda)_j^-,\phi(\lambda)_j^+\big\}_{j\in\N(\lambda)}$ for
almost every $\lambda\in[0,\infty)$, one has
$$
F_0^{-1}\phi=\big(\F^{-1}\otimes1\big)\widetilde\phi\quad\hbox{with}\quad
\widetilde\phi(\xi,\;\!\cdot\;\!):=
\begin{cases}
\sqrt{2|\xi|}\sum_{j\in\N}\phi(\xi^2+\tau_j)_j^- & \hbox{for almost every}~~\xi<0\\
\sqrt{2|\xi|}\sum_{j\in\N}\phi(\xi^2+\tau_j)_j^+ & \hbox{for almost every}~~\xi\ge0.
\end{cases}
$$
\end{enumerate}
\end{Lemma}

\begin{proof}
Point (a) can be shown as in Lemma 2.4.(a) of \cite{Tie06}. For (b), a look at the
expression \eqref{eq_T_0} for $T_0(\lambda)$ shows it is sufficient to prove that the
function $\gamma:\R\to\B\big(\H^s_t(\R),\C\big)$ is $k$-times H\"older continuously
differentiable. But, we already know from \cite[Lemma~A.1]{Tie09} that $\gamma$ is
$k$-times H\"older continuously differentiable from $\R$ to
$\B\big(\H^s(\R),\C\big)$. This fact, together with the identity
$$
\gamma(\xi)\langle Q\rangle^t\varphi=\langle\xi\rangle^t\gamma(\xi)\varphi,
\quad\varphi\in\H^s(\R),~\xi,t\in\R,
$$
implies the desired differentiability.

Finally, the result of point (c) can be shown as in Proposition 2.5 of \cite{Tie06}.
\end{proof}

\section{Manifold with asymptotically cylindrical ends}\label{secfull}
\setcounter{equation}{0}

Let $(M,g)$ be a smooth, second countable, complete Riemannian manifold of dimension
$n+1$, without boundary. Assume that $M$ is of the form $M=\Mc\cup\Minf$, with $\Mc$
relatively compact and $\Minf$ open in $M$. Moreover, suppose that $\Minf$ (with the
induced atlas) can be identified to $(0,\infty)\times\Sigma$ (with the direct product
atlas) in the following sense: There exists a diffeomorphism
$\iota:\Minf\to(0,\infty)\times\Sigma$ mapping each local chart of $\Minf$ to a local
chart of $(0,\infty)\times\Sigma$. In other terms, if the collection
$\{(\V_\alpha,\rho_\alpha)\}$ stands for the atlas on $\Minf$, then the collection
$$
\big\{\big(\U_\alpha,(x,\omega_\alpha)\big)\big\}
:=\big\{\big(\iota(\V_\alpha),\rho_\alpha\circ\iota^{-1}\big)\big\}
$$
defines an equivalent atlas on $(0,\infty)\times\Sigma$. We also assume that
$\iota(\Mc\cap\Minf)\subset(0,1)\times\Sigma$.

We denote by $g_{jk}$ the components of $g$ on a chart $(W,\zeta)$ of $M$, we set
$\{g^{jk}\}:=\{g_{jk}\}^{-1}$, and we define the volume element $\dv$ on $M$
as
$$
\dv:={\frak g}\,\d\zeta\quad\hbox{with}\quad{\frak g}:=\sqrt{\det\{g_{jk}\}}\;\!.
$$
In the Hilbert space $\H:=\ltwo(M,\dv)$ we consider the operator $H$ given by
$$
H\psi:=\big(\Delta_M+V\big)\psi,\quad\psi\in C^\infty_{\rm c}(M),
$$
where $\Delta_M$ is the (Dirichlet) Laplace-Beltrami operator on $M$ and $V$ belongs
to the set $C_{\rm b}^\infty(M)$ of smooth functions on $M$ with all derivatives
bounded (note that we use the same notation for a function and for the corresponding
multiplication operator). Since $M$ is complete and $V$ is bounded, the operator $H$
is essentially self-adjoint on $C^\infty_{\rm c}(M)$ \cite[Thm.~3]{Cor72}, and $H$
acts as
$$
H\psi:=-\sum_{j,k=1}^{n+1}{\frak g}^{-1}
\frac\partial{\partial\zeta^j}\;\!{\frak g}\;\!g^{jk}
\frac\partial{\partial\zeta^k}\;\!\psi+V\psi,
\quad\psi\in C_{\rm c}^\infty(M),
$$
on each chart $(W,\zeta)$ of $M$. Now, on each complete Riemannian manifold $\cal M$, one can
define the Sobolev spaces $W^k(\cal M)$ given in terms of the covariant derivatives
and the Sobolev spaces $\H^{2k}(\cal M)$ given in terms of the Laplace-Beltrami
operator (see \cite[Sec.~0]{Sal01}).
Therefore, the domain $\dom(H)$ of $H$ satisfies in our situation
$$
\dom(H)=\dom(\Delta_M)=\H^2(M).
$$

In the next two lemmas, we recall a compacity criterion and a result on elliptic
regularity that will be used in various instances. In the first lemma, the set of
continuous bounded functions on $M$ is denoted by $C_{\rm b}(M)$, and the ideal of
compact operators of $\B(\H)$ is denoted by $\K(\H)$. In the second lemma, the
Laplacian $\Delta_M$ contained in $H\equiv\Delta_M+V$ is regarded as the
distributional Laplacian on $\ltwoloc(M)$ (the distributional Laplacian coincides
with the usual Dirichlet Laplacian on the subset $\H^2(M)\subset\ltwoloc(M)$, see
\cite[Sec.~4.1-4.2]{Gri09} for details).

\begin{Lemma}\label{lemmecompact}
Let $m\in C_{\rm b}(M)$ satisfy
$
\lim_{x\to\infty}
\big\|\big(m\circ\iota^{-1}\big)(x,\;\!\cdot\;\!)\big\|_{\linf(\Sigma)}=0.
$
Then the product $m(H\pm i)^{-1}$ belongs to $\K(\H)$.
\end{Lemma}

\begin{proof}
Let $\V$ be any open relatively compact subset of $M$ and let $\chi_\V$ denote the
corresponding characteristic function. Then, one shows using standards results (see
\cite[Sec.~1.2]{Ura93}, \cite[Sec.~2.2]{Gri99} and \cite[Sec.~1]{Sal01}) that the
operators $\chi_\V(\Delta_M\pm i)^{-1}$ belong to $\K(\H)$. Since $\K(\H)$ is closed
in the norm topology, one infers by an approximation argument taking the geometry of
$M$ into account that $m(\Delta_M\pm i)^{-1}\in\K(\H)$. One then concludes by using
the second resolvent equation
$(H\pm i)^{-1}=(\Delta_M\pm i)^{-1}\big\{1-V(H\pm i)^{-1}\big\}$.
\end{proof}

\begin{Lemma}[Elliptic regularity]\label{ellipticite}
Assume that $V$ is bounded from $\H^{2\ell}(M)$ to $\H^{2\ell}(M)$ for each
$\ell\in\N$. Let $u\in\ltwoloc(M)$ satisfy $(H-z)u=f$ for some
$z\in\C\setminus\sigma(H)$ and $f\in C^\infty(M)$. Then, $u\in C^\infty(M)$.
\end{Lemma}

Note that the boundedness of $V$ from $\H^{2\ell}(M)$ to $\H^{2\ell}(M)$ is
automatically verified if the curvature tensor of $M$ and all its derivatives are
uniformly bounded (see the next section).

\begin{proof}
We know from \cite[Cor.~7.2]{Gri09} that the identity map
$I:C^\infty(M)\to\H^\infty_{\rm loc}(M)$ is a homeomorphism of topological spaces.
Therefore, it is sufficient to show that $(\Delta_M)^ku\in\ltwoloc(M)$ for each
$k\in\N$. We proceed by induction on $k$. For $k=1$, we have
$\Delta_Mu=f+zu-Vu\in\ltwoloc(M)$, since each term belongs to $\ltwoloc(M)$. For
$k-1\geq 0$, we assume the statement true. Then, for $k$, we have
$$
(\Delta_M)^ku=(\Delta_M)^{k-1}f+z(\Delta_M)^{k-1}u-(\Delta_M)^{k-1}Vu.
$$
But, $(\Delta_M)^{k-1}f\in\ltwoloc(M)$ since $f\in C^\infty(M)$ and
$(\Delta_M)^{k-1}u\in\ltwoloc(M)$ due to the induction hypothesis. So, it only
remains to show that $(\Delta_M)^{k-1}Vu\in\ltwoloc(M)$. However, we know by
assumption that $V$ is bounded from $\H^{2\ell}(M)$ to $\H^{2\ell}(M)$ for any
$\ell\in\N$. Therefore, $V$ is also bounded from $\H^{2(k-1)}(\V)$ to
$\H^{2(k-1)}(\V)$ for any relatively compact open set $\V\subset M$.
This, together with the induction hypothesis $u\in\H^{2(k-1)}(M)$
implies that $(\Delta_M)^{k-1}Vu\in\ltwoloc(M)$.
\end{proof}

\section{Spectral analysis}\label{sec_spect}
\setcounter{equation}{0}

We perform in this section the spectral analysis of the operator $H$. We impose
explicit decay assumptions on the metric and on the potential at infinity. Then, we
deduce various results on the regularity of the resolvent of $H$ near the real axis.
For that purpose, we use Mourre theory in the rather efficient way presented in the
paper \cite{RT12}; namely, we build the Mourre theory for $H$ from the analog theory
for $H_0$, even if $H$ and $H_0$ act in different Hilbert spaces.

To begin with, we need to introduce an identification operator from $\H_0$ to $\H$.
For this, we recall that $\iota^*$ and $(\iota^{-1})_*$ are respectively the pullback
by $\iota$ and the push-forward by $\iota^{-1}$. Then, we let
$\j\in C^\infty\big(\R;[0,1]\big)$ satisfy
$$
\j(x):=
\begin{cases}
1 & \hbox{if}~~x\ge2\\
0 & \hbox{if}~~x\le1,
\end{cases}
$$
and set
\begin{equation*}
J:\H_0\to\H,\quad\varphi\mapsto
\chi_\infty\sqrt{\frac{\iota^*(1\otimes{\frak h})}{\frak g}}\,
\iota^*\big((\j\otimes1)\varphi\big),
\end{equation*}
where $\chi_\infty$ is the characteristic function for $\Minf$ and
${\frak h}:=\sqrt{\det\{h_{jk}\}}$. One has $\|J\|_{\B(\H_0,\H)}= 1$, since
$\|J\varphi\|_\H=\|\varphi\|_{\H_0}$ for each function $\varphi\in\H_0$ with
$\supp(\varphi)\subset(2,\infty)\times\Sigma$.

Our second task consist in fixing the decay behaviour of the metric and the potential
on $M_\infty$. In our setup, conditions on $g_{jk}$, $V$ and $\iota$ could be stated
separately. We prefer to combine these conditions in a single one on
$\widetilde g_{jk}:=(\iota^{-1})^*g_{jk}$ and $\widetilde V:=(\iota^{-1})^*V$, since
it corresponds to the usual approach in the literature. Note that even if we work in
a smooth setting, we shall distinguish short-range and long-range behaviours for the
sake of completeness.

In the following assumption, we use the notation $\partial^\alpha$ for the higher
order derivative $(\partial_x,\partial_\omega)^\alpha$ with multi-index
$\alpha\in\N^{n+1}$ and the notation $\langle x\rangle:=(1+x^2)^{1/2}$ for $x\in\R$.

\begin{Assumption}\label{CondsurgV}
Assume that the metric
$\widetilde g_{jk}=(\widetilde g_{\rm L})_{jk}+(\widetilde g_{\rm S})_{jk}$ and the
potential $\widetilde V=\widetilde V_{\rm L}+\widetilde V_{\rm S}$ satisfy the
following:
\begin{enumerate}
\item[(LR)] There exists $\muL\ge0$ such that for each $\alpha\in\N^{n+1}$ and each
$j,k\in\{1,\ldots,n+1\}$ one has
$$
\big|\partial^\alpha\big((\widetilde g_L)_{jk}-(1\oplus h)_{jk}\big)(x,\omega)\big|
\le\textsc c_\alpha\;\!\langle x\rangle^{-\muL-|\alpha|}\quad\hbox{and}\quad
\big|\big(\partial^\alpha\widetilde V_{\rm L}\big)(x,\omega)\big|
\le\textsc c_\alpha\;\!\langle x\rangle^{-\muL-|\alpha|}
$$
for some constant $\textsc c_\alpha\ge0$ and for all $x>0$ and $\omega\in\Sigma$.
\item[(SR)] There exists $\muS\ge0$ such that for each $\alpha\in\N^{n+1}$ and each
$j,k\in\{1,\ldots,n+1\}$ one has
$$
\big|\partial^\alpha\big((\widetilde g_S)_{jk}-(1\oplus h)_{jk}\big)(x,\omega)\big|
\le\textsc c_\alpha\;\!\langle x\rangle^{-\muS}\quad\hbox{and}\quad
\big|\big(\partial^\alpha\widetilde V_{\rm S}\big)(x,\omega)\big|
\le\textsc c_\alpha\;\!\langle x\rangle^{-\muS}
$$
for some constant $\textsc c_\alpha\ge0$ and for all $x>0$ and $\omega\in\Sigma$.
\end{enumerate}
Let also $\mu:=\min\{\muL,\muS\}$.
\end{Assumption}

Simple consequences of Assumption \ref{CondsurgV} are the following:

\begin{enumerate}
\item[(i)] For each $\alpha\in\N^{n+1}$, one has
$
\big|(\partial^\alpha\widetilde g_{jk})(x,\omega)\big|
\le\textsc d_\alpha
$
for some constant $\textsc d_\alpha\ge0$ and for all $x>0$ and $\omega\in\Sigma$.
\item[(ii)] There exists a constant $\delta>0$ such that $\widetilde{\frak g}>\delta$
on $(0,\infty)\times\Sigma$.
\item[(iii)] The curvature tensor of $M$ is uniformly bounded, as are all its
covariant derivatives. In particular, the Sobolev spaces $W^{2k}(M)$ and $\H^{2k}(M)$
are equal for all $k\in\N$ and $V$ is bounded from $\H^{2\ell}(M)$
to $\H^{2\ell}(M)$ for any $\ell\in\N$ (see \cite[Sec.~5~\&~Lemma~1.6]{Sal01}). So, $\dom(H)=\dom(\Delta_M)=\H^2(M)=W^2(M)$ and Lemma \ref{ellipticite} applies.
\end{enumerate}

Now, we determine an expression for the operator $HJ-JH_0$ acting on suitable
elements of $\H_0$. The main ingredient of the computation is the following equality
\begin{equation}\label{magic}
\frac\partial{\partial\rho^j}\;\!\psi
=\iota^*\left\{\frac\partial{\partial(x,\omega)^j}\big((\iota^{-1})^*\psi\big)\right\},
\quad\supp(\psi)\subset\Minf\;\!,
\end{equation}
which follows from the definition of the diffeomorphism $\iota$.
Using the matricial conventions, we obtain for any
$\varphi\in\S(\R)\odot C^\infty(\Sigma)$ that
\begin{align}\label{defT0}
T\varphi
&:=(HJ-JH_0)\varphi\nonumber\\
&=-\chi_\infty\;\!\iota^*\bigg[\widetilde{\frak g}^{-1}(1\otimes\frak h)^{-1/2}
\;\!b_1\big(\partial_x,\partial_\omega\big)\widetilde{\frak g}\;\!\widetilde g^{-1}
\begin{pmatrix}
\partial_x\nonumber\\
\partial_\omega
\end{pmatrix}
\widetilde{{\frak g}}^{-1/2}(1\otimes\frak h)^{1/2}(\j \otimes 1)\varphi\nonumber\\
&\quad+\widetilde{\frak g}^{-1/2}(1\otimes\frak h)^{-1/2}(\j\otimes 1)
\big(\partial_x,\partial_\omega\big)b_2
\begin{pmatrix}
\partial_x\\
\partial_\omega
\end{pmatrix}
\widetilde{{\frak g}}^{-1/2}(1\otimes\frak h)^{1/2}(\j\otimes1)\varphi\nonumber\\
&\quad+\widetilde{\frak g}^{-1/2}(1\otimes\frak h)^{-1/2}(\j\otimes 1)
\big(\partial_x,\partial_\omega\big)(1\otimes\frak h)\big(1\oplus h^{-1}\big)
\begin{pmatrix}
\partial_x\\
\partial_\omega
\end{pmatrix}
b_3\;\!\widetilde{\frak g}^{-1/2}\varphi\nonumber\\
&\quad-\widetilde V\;\!\widetilde{{\frak g}}^{-1/2}(1\otimes\frak h)^{1/2}
(\j\otimes1)\varphi\bigg],
\end{align}
with
$b_1:=(1\otimes\frak h)^{1/2}-\widetilde{\frak g}^{1/2}(\j\otimes 1)$,
$
b_2:=\widetilde{\frak g}\;\!\widetilde g^{-1}-(1\otimes\frak h)
\big(1\oplus h^{-1}\big)
$
and $b_3:=(1\otimes\frak h)^{1/2}(\j\otimes 1)-\widetilde{\frak g}^{1/2}$.

The following lemma will be used at various places in the sequel. Its statement
involves the multiplication operator $\Phi_0$ on $\R\times\Sigma$ given by
\begin{equation}\label{defphi0}
\Phi_0\varphi:=(\id\otimes1)\varphi,\quad\varphi\in\S(\R)\odot C^\infty(\Sigma),
\end{equation}
where $\id$ is the function $\R\ni x\mapsto x\in\R$. The closure of $\Phi_0$ in
$\H_0$ (which we denote by the same symbol) is self-adjoint.

\begin{Lemma}\label{preuvepourRafa}
Suppose that Assumption \ref{CondsurgV} holds with $\mu\ge0$ and take
$\gamma\in[0,\mu]$. Then, the operator $T\langle\Phi_0\rangle^\gamma$ defined on
$\S(\R)\odot C^\infty(\Sigma)$ extends continuously to an element of
$\B\big(\dom(H_0),\H\big)$. Furthermore, for any $z\in\C\setminus\R$ the operator
$(H-z)^{-1}T\langle\Phi_0\rangle^\gamma$ defined on $\S(\R)\odot C^\infty(\Sigma)$
extends continuously to an element of $\B(\H_0,\H)$.
\end{Lemma}

\begin{proof}
We know that
$
T\langle\Phi_0\rangle^\gamma\varphi
=-\chi_\infty\;\!\iota^*\big(T^0\langle\Phi_0\rangle^\gamma\varphi\big)
$
for any $\varphi\in\S(\R)\odot C^\infty(\Sigma)$, where $T^0$ is the differential
operator within the square brackets in \eqref{defT0}. Furthermore, some routine
computations involving Assumption \ref{CondsurgV} and its consequences (i) and (ii)
imply that for each $\alpha\in\N^{n+1}$ there exists a constant
$\textsc d_\alpha\ge0$ such that
\begin{equation}\label{b_bound}
\textstyle\big\{\big|(\partial^\alpha b_1)(x,\omega)\big|
+\sum_{j,k=1}^{n+1}\big|[\partial^\alpha(b_2)_{jk}](x,\omega)\big|
+\big|(\partial^\alpha b_3)(x,\omega)\big|\big\}\langle x\rangle^\mu
\le\textsc d_\alpha
\end{equation}
for all $x>0$ and $\omega\in\Sigma$. Therefore, the operator
$T^0\langle\Phi_0\rangle^\gamma$ is a second order differential operator on
$\S(\R)\odot C^\infty(\Sigma)$ with coefficients in
$C_{\rm b}^\infty(\R\times\Sigma)$. So, it follows from \cite[Lemma~1.6]{Sal01} that
$T^0\langle\Phi_0\rangle^\gamma$ extends continuously to a bounded operator (denoted
similarly) from $W^2(\R\times\Sigma)$ to $W^0(\R\times\Sigma)\equiv\H_0$. Now, since
$\R\times\Sigma$ is geodesically complete and with bounded curvatures, one also has
\cite[Sec.~5]{Sal01} $W^2(\R\times\Sigma)=\H^2(\R\times\Sigma)\equiv\dom(H_0)$, and
thus $T^0\langle\Phi_0\rangle^\gamma$ extends to a bounded operator from $\dom(H_0)$
to $\H_0$. This result, together with the inclusion
$\chi_\infty\;\!\iota^*\in\B(\H_0,\H)$, implies the first statement.

For the second statement, we consider for $\psi\in(H-\bar z)\;\!C^\infty_{\rm c}(M)$
and $\varphi\in\S(\R)\odot C^\infty(\Sigma)$ the equality
$$
\big\langle\psi,(H-z)^{-1}T\langle\Phi_0\rangle^\gamma\varphi\big\rangle_\H
=\big\langle\langle\Phi_0\rangle^\gamma T^*(H-\bar z)^{-1}\psi,
\varphi\big\rangle_{\H_0}.
$$
Furthermore, for any $\zeta\in C^\infty_{\rm c}(M)$, we observe that
$
\langle\Phi_0\rangle^\gamma T^*\zeta
=\chi_{(0,\infty)\times\Sigma}(\iota^{-1})^*(L\zeta)
$,
where $L$ is a second order differential operator on $C^\infty_{\rm c}(M)$ with
coefficients in $C_{\rm b}^\infty(M)$. Now, we know from \cite[Lemma~1.6]{Sal01} and
the consequence (iii) of Assumption \ref{CondsurgV} that $L$ extends continuously to
a bounded operator from $W^2(M)\equiv\dom(H)$ to $\H$. Thus, the statement follows
from the density of $(H-\bar z)\;\!C^\infty_{\rm c}(M)$ in $\H$ and the density of
$\S(\R)\odot C^\infty(\Sigma)$ in $\H_0$.
\end{proof}

Let us finally note that the previous result implies in particular that
$J\in\B\big(\dom(H_0),\dom(H)\big)$.

\subsection{Conjugate operator for $\boldsymbol H$}\label{Sec_conjugate}

In this section, we define a conjugate operator for $H$ and use it to deduce some
standard results. The conjugate operator could be either defined as a geometric
object or as a modification of the generator of dilations on $\R$. We present the
former approach because self-adjointness is automatically obtained, but we link
afterward the two possible constructions.

So, let $X\in\X(M)$ be the smooth vector field defined by
$$
X:=\chi_\infty\;\!\iota^*\big(\j^2\;\!\id\otimes1\big)\;\!
(\iota^{-1})_*\left(\frac\partial{\partial x}\right).
$$
Given $p\in M$, it is known \cite[Sec.~2.1]{AM78} that there exist $\varepsilon>0$, a
neighbourhood $\V\subset M$ of $p$ and a smooth map
$F:(-\varepsilon,\varepsilon)\times\V\to M$ satisfying for each
$(\tau,q)\in(-\varepsilon,\varepsilon)\times \V$ the differential equation
$\frac\d{\d\tau}F(\tau,q)=X_{F(\tau,q)}$, $F(0,q)=q$. Furthermore, for each
$\tau\in(-\varepsilon,\varepsilon)$ the map $F_\tau:=F(\tau,\;\!\cdot\;\!)$ is a
diffeomorphism onto its image. In fact, one has $F(\tau,p)=p$ for all
$(\tau,p)\in\R\times M\setminus\Minf$ since $X\equiv0$ on $M\setminus\Minf$, and one
can show that the vector field $X$ is complete by applying the criterion
\cite[Prop.~2.1.20]{AM78} with the proper function $f:M\to\R$ given by
$f:=\chi_\infty\;\!\iota^*\big(\j^2\;\!\id\otimes1\big)$. So, the restricted map
$F_\tau|_\Minf:\Minf\to\Minf$ is a diffeomorphism for each $\tau\in\R$.

Based on the complete vector field $X$ one can construct a unitary group acting on
$\H$. However, $M$ being a priori not orientable, one has to take some extra care
when defining the group: Since the manifolds $\Sigma_\ell$ are orientable, it follows
that $\Minf\equiv\iota^{-1}\big((0,\infty)\times\Sigma\big)$ is also orientable. So,
$\dv$ is a volume form on $\Minf$ \cite[Thm.~7.7]{Boo86}, and there exists a unique
smooth function $\det_\dv(F_\tau|_\Minf):M\to\R$, called the determinant of
$F_\tau|_\Minf$, which satisfies
$\big(F_\tau|_\Minf\big)^*\dv=\det_\dv(F_\tau|_\Minf)\;\!\dv$
\cite[Def.~2.5.18]{AM78}. For each $\tau\in\R$ we can thus define the map
$$
J_\tau:M\to\R,\quad p\mapsto
\begin{cases}
1 & \hbox{if}~~p \in M\setminus\Minf\\
\det_\dv(F_\tau|_\Minf)(p) & \hbox{if}~~p\in\Minf.
\end{cases}
$$
Since $F_\tau|_{\Mc\cap\Minf}$ is the identity map, we have
$\det_\dv(F_\tau|_\Minf)=1$ on $\Mc\cap\Minf$ \cite[Prop.~2.5.20.(ii)]{AM78}, and
thus $J_\tau$ is a smooth function on $M$.

We can now define for each $\tau\in\R$ and each $\psi\in C_{\rm c}^\infty(M)$ the
operator
$$
U(\tau)\;\!\psi:=J_\tau^{1/2}F_\tau^*\psi.
$$
Some routine computations using \cite[Prop.~2.5.20]{AM78} show that $U(\tau)$ can be
extended to an isometry from $\H$ to $\H$ (which we denote by the same symbol), and
that $\{U(\tau)\}_{\tau\in\R}$ defines a strongly continuous unitary group in $\H$.
Furthermore, since $J_\tau(p)>0$ for all $p \in M$, one sees easily that
$U(\tau)C_{\rm c}^\infty(M)\subset C_{\rm c}^\infty(M)$. Thus, one can apply Nelson's
Lemma to show that the generator $A$ of the unitary group $\{U(\tau)\}_{\tau\in\R}$
is essentially self-adjoint on $C_{\rm c}^\infty(M)$. Direct computations with
$\psi\in C_{\rm c}^\infty(M)$ (see \cite[Sec.~5.4]{AM78}) show that
\begin{align}
A\psi&=-i\chi_\infty\big(\L_X+\12\;\!\div_\dv X\big)\psi,\label{A_geo}\\
\L_X\psi&=\textstyle\chi_\infty\;\!\iota^* \big\{(\j^2\;\!\id\otimes 1)
\frac\partial{\partial x}(\iota^{-1})^*\psi\big\},\nonumber\\
\div_\dv X&=\frak g^{-1}\L_X\frak g
+\iota^*\big\{\big((\j^2)'\;\!\id+\j^2\big)\otimes1\big\},\nonumber
\end{align}
with $\L_X$ the Lie derivative along $X$ and $\div_\dv X$ the divergence of $X$ with
respect to the volume form $\dv$ of $\Minf$. Note that the function
$\chi_\infty\;\!\div_\dv X$ belongs to $C^\infty_{\rm b}(M)$ under Assumption
\ref{CondsurgV} with $\muL\ge0$ and $\muS\ge1$.

\begin{Remark}\label{SurA}
Let $A_0$ be the generator of dilations in $\H_0$, that is, the operator given by
$A_0:=\12(PQ+QP)\otimes1$. Then a direct calculation
shows that
\begin{equation*}
A\psi=JA_0J^*\psi
\end{equation*}
for any $\psi\in C_{\rm c}^\infty(M)$. Therefore, the operator $A$ is nothing else
but the generator of dilations $A_0$ injected in the Hilbert space $\H$ via the
identification operator $J$.
\end{Remark}

We can now study the regularity of the operators $H_0$ and $H$ with respect to the
operators $A_0$ and $A$. For this we mainly use the framework and notations from
\cite{ABG}. In particular, we say that the self-adjoint operator $H$ is of class
$C^k(A)$, $k\in\N$, if the map
\begin{equation*}
\R\ni t\mapsto\e^{-itA}(H-i)^{-1}\e^{itA}\in\B(\H)
\end{equation*}
is $k$-times strongly differentiable. In the case of a bounded operator $B\in\B(\H)$,
this is equivalent to showing that the map $ t\mapsto\e^{-itA}B\e^{itA}$ is $k$-times
strongly differentiable, and we write $B\in C^k(A)$. The same definitions hold with
$\H,H,A$ replaced by $\H_0,H_0,A_0$. Due to its simplicity and its tensorial
structure, it is easily shown that $H_0$ is of class $C^k(A_0)$, with $A_0$ defined
in Remark \ref{SurA}, for any $k\in\N$. In the next lemma, whose proof is inspired
from \cite[Sec.~2.1]{Bou06} and \cite[Lemma A.2]{GM08}, we show that $H$ is of class $C^1(A)$
(higher regularity of $H$ with respect to $A$ will be considered in Section
\ref{SecHighR}). As mentioned in the Appendix A of \cite{GM08}, checking the
$C^1(A)$-condition is sometimes omitted in the Mourre analysis on a manifold, and
without this condition the application of the Virial Theorem is erroneous.

\begin{Lemma}\label{regul}
Suppose that Assumption \ref{CondsurgV} holds with $\muL\ge0$ and $\muS\ge1$. Then
$H$ is of class $C^1(A)$.
\end{Lemma}

\begin{proof}
Consider the family of multiplication operators $\chi_n\in\B(\H)$ defined as follows:
Let $\eta\in C^\infty(\R;\R)$ satisfy $\eta(x)=1$ if $x\le1$ and $\eta(x)=0$ if
$x\ge2$, and for any $n\in\N^*$ let $\chi_n\in C^\infty_{\rm c}(M;\R)$ be given by
$\chi_n=1$ on $M\setminus M_\infty$ and $[(\iota^{-1})^*\chi_n](x,\omega):=\eta(x/n)$
for $(x,\omega)\in(0,\infty)\times\Sigma$.

Then, one has $\slim_{n\to\infty}\chi_n=1$, and a direct calculation taking Remark
\ref{SurA} into account shows that $\lim_{n\to\infty}A\chi_n\psi=A\psi$ for each
$\psi\in C_{\rm c}^\infty(M)$. Furthermore, Lemma \ref{ellipticite} implies that
$\chi_n(H+i)^{-1}C_{\rm c}^\infty(M)\subset C_{\rm c}^\infty(M)$, and lengthy but
standard computations involving the identity \eqref{magic} show that
$\lim_{n\to\infty}A[H,\chi_n](H+i)^{-1}\psi=0$ for each
$\psi\in C_{\rm c}^\infty(M)$. Using these facts, one obtains that
\begin{align*}
\big\langle(H-i)^{-1}\psi,A\psi\big\rangle_\H
-\big\langle A\psi,(H+i)^{-1}\psi\big\rangle_\H
&=\lim_{n\to\infty}\big\langle\psi,\big[(H+i)^{-1},A\chi_n\big]
\psi\big\rangle_\H\\
&=\lim_{n\to\infty}\big\langle\psi,-(H+i)^{-1}[H,A]
\chi_n(H+i)^{-1}\psi\big\rangle_\H.
\end{align*}
Now, a routine computation taking into account Formula \eqref{A_geo}, Assumption
\ref{CondsurgV} with $\muL\ge0$ and $\muS\ge 1$, and the bound \eqref{b_bound} shows
that there exists a second order differential operator $L$ with coefficients in
$C^\infty_{\rm b}(M)$ such that $[H,A]=L$ on $C_{\rm c}^\infty(M)$. Since $L$ extends
continuously to a bounded operator from $W^2(M)\equiv\dom(H)$ to $\H$ due to
\cite[Lemma~1.6]{Sal01}, one obtains that
\begin{equation}\label{udon}
\big\langle(H-i)^{-1}\psi,A\psi\big\rangle_\H
-\big\langle A\psi,(H+i)^{-1}\psi\big\rangle_\H
=\big\langle\psi,-(H+i)^{-1}L(H+i)^{-1}\psi\big\rangle_\H.
\end{equation}
But, the set $C^\infty_{\rm c}(M)$ is a core for $A$, thus \eqref{udon} even holds
for $\psi\in\dom(A)$. So, the quadratic form
$
\dom(A)\ni\psi\mapsto\big\langle(H-i)^{-1}\psi,A\psi\big\rangle_\H
-\big\langle A\psi,(H+i)^{-1}\psi\big\rangle_\H
$
extends uniquely to the bounded form defined by the operator
$-(H+i)^{-1}L(H+i)^{-1}\in\B(\H)$, and thus $H$ is of class $C^1(A)$ (see
\cite[Def.~6.2.2]{ABG}).
\end{proof}

Lemma \ref{regul} implies in particular that $\Delta_M$ is of class $C^1(A)$, since
the potential $V=0$ satisfies Assumption \ref{CondsurgV} for any $\muL,\muS\ge0$. To
close the section, we show that the group $\{\e^{itA}\}_{t\in\R}$ leaves the domain
$\dom(H)\equiv\dom(\Delta_M)$ invariant:

\begin{Lemma}\label{invariance}
Suppose that Assumption \ref{CondsurgV} holds with $\muL\ge0$ and $\muS\ge 1$. Then
$\e^{itA}\dom(H)\subset\dom(H)$ for all $t\in\R$.
\end{Lemma}

\begin{proof}
As mentioned in the previous proof there exists a second order differential operator
$L\in\B\big(\dom(H),\H\big)$ such that $[H,A]=L$ on $C_{\rm c}^\infty(M)$.
So, Lemma \ref{regul} together with \cite[Eq.~6.2.24]{ABG}
imply, in the form sense on $\H$, that
$$
(H+i)^{-1}[H,A](H+i)^{-1}=(H+i)^{-1}L(H+i)^{-1},
$$
where $[H,A]\in\B\big(\dom(H),\dom(H)^*\big)$ is the operator
associated with the unique extension to $\dom(H)$ of the quadratic form
$
\dom(H)\cap\dom(A)\ni\psi\mapsto
\langle H\psi,A\psi\rangle_\H
-\langle A\psi,H\psi\rangle_\H
$.
Therefore, $L$ and $[H,A]$ are equals in $\B\big(\dom(H),\dom(H)^*\big)$, and
$[H,A]\;\!\dom(H)=L\;\!\dom(H)\subset\H$. The claim then follows from Lemma
\ref{regul} and \cite[Lemma~2]{GG99}.
\end{proof}

\subsection{Mourre estimate}

In reference \cite{RT12}, an abstract method giving a Mourre estimate for $H$ from a
Mourre estimate for the pair $(H_0,A_0)$ has been developed. The verification of the
assumptions necessary to apply this method is the content of the next lemmas. Here,
$C_0(\R)$ denotes the set of continuous functions on $\R$ vanishing at $\pm\infty$.

\begin{Lemma}\label{compacite1}
Suppose that Assumption \ref{CondsurgV} holds with $\mu>0$, and let
$\eta\in C_0(\R)$. Then the difference $J\eta(H_0)-\eta(H)J$ belongs to
$\K(\H_0,\H)$.
\end{Lemma}

\begin{proof}
Let $z\in\C\setminus\R$. We know from Lemma \ref{preuvepourRafa} that
$(H-z)^{-1}T\langle\Phi_0\rangle^\mu$, defined on $\S(\R)\odot C^\infty(\Sigma)$,
extends continuously to an operator $C(z)\in\B(\H_0,\H)$. Furthermore, one can show
by mimicking the proof of \cite[Lemma~2.1]{KR04} that
$K(z):=\langle\Phi_0\rangle^{-\mu}(H_0-z)^{-1}$ belongs to $\K(\H_0)$. So, one has on
$(H_0-z)\big(\S(\R)\odot C^\infty(\Sigma)\big)$ the equalities
$$
J(H_0-z)^{-1}-(H-z)^{-1}J
=(H-z)^{-1}T\langle\Phi_0\rangle^\mu\langle\Phi_0\rangle^{-\mu}(H_0-z)^{-1}
=C(z)K(z),
$$
and by the density of $(H_0-z)\big(\S(\R)\odot C^\infty(\Sigma)\big)$ in $\H_0$ these
equalities extend continuously to $\H_0$. One concludes by taking into account the
fact that the vector space generated by the family of functions
$\{(\;\cdot\;-z)^{-1}\}_{z\in\C\setminus\R}$ is dense in $C_0(\R)$ and that the set
$\K(\H_0,\H)$ is closed in $\B(\H_0,\H)$.
\end{proof}

Suppose that Assumption \ref{CondsurgV} holds with $\muL\ge0$ and $\muS\ge1$, and let
$\eta\in C^\infty_{\rm c}(\R)$. Then, we deduce from Lemma \ref{regul} and
\cite[Thm.~6.2.5]{ABG} that $\eta(H)\in C^1(A)$. Therefore, the quadratic form
$
\dom(A)\ni\psi\mapsto\langle A\psi,\eta(H)\psi\rangle_\H
-\langle\bar\eta(H)\psi,A\psi\rangle_\H
$
extends uniquely to a bounded form on $\H$, with corresponding bounded operator
denoted by $[A,\eta(H)]$. Since, the same holds for the pair $(H_0,A_0)$ in $\H_0$,
one can define similarly the operator $[A_0,\eta(H_0)]\in\B(\H_0)$.

The next lemma shows that these two commutators do not differ too much, even though
they live in different Hilbert spaces:

\begin{Lemma}\label{C2}
Suppose that Assumption \ref{CondsurgV} holds with $\muL>0$ and $\muS>1$, and let
$\eta\in C^\infty_{\rm c}(\R)$. Then, the difference of bounded operators
$J[A_0,\eta(H_0)]J^*-[A,\eta(H)]$ belongs to $\K(\H)$.
\end{Lemma}

\begin{proof}
We apply \cite[Prop.~3.12]{RT12}, which shows in an abstract framework how the
inclusion $J[A_0,\eta(H_0)]J^*-[A,\eta(H)]\in \K(\H)$ follows from a certain set of
hypotheses. Therefore, we simply check the hypotheses in question.

First, we know that $H_0$ is of class $C^1(A_0)$ with
$[H_0,A_0]=-2iP^2\otimes1\in\B\big(\dom(H_0),\H_0\big)$ and that $H$ of class
$C^1(A)$ due to Lemma \ref{regul}. Next, one has to show that the operator $J$
extends to an element of $\B\big(\dom(H_0)^*,\dom(H)^*\big)$. For this, let
$\D:=\{\varphi\in\S(\R)\odot C^\infty(\Sigma)\mid\|\varphi\|_{\H_0}=1\}$ and observe
that
\begin{align*}
\|J\|_{\B(\dom(H_0)^*,\dom(H)^*)}
&\le{\rm Const.}\;\!\big\|\langle H\rangle^{-1}J(H_0+1)\big\|_{\B(\H_0,\H)}\\
&\le{\rm Const.}\;\!
\bigg(1+\sup_{\varphi\in\D}\big\|\langle H\rangle^{-1}JH_0\varphi\big\|_\H\bigg)\\
&={\rm Const.}\;\!
\bigg(1+\sup_{\varphi\in\D}\big\|\langle H\rangle^{-1}(HJ-T)\varphi\big\|_\H\bigg)
\end{align*}
which is finite due to Lemma \ref{preuvepourRafa}.

Two additional hypotheses have to be checked. The first one is the inclusion
$J(H_0-z)^{-1}-(H-z)^{-1}J\in\K(\H_0,\H)$, $z\in\C\setminus\R$, which has already
been obtained in the proof of Lemma \ref{compacite1}. The second one is the inclusion
$J[H_0,A_0]J^*-[H,A]\in\K\big(\dom(H),\dom(H)^*\big)$ (note that we already know that
$J[H_0,A_0]J^*-[H,A]$ is bounded from $\dom(H)$ to $\dom(H)^*$ due to the previous
observations). Now, a rather lengthy but straightforward computation taking
Assumption \ref{CondsurgV} into account shows for all
$\varphi\in C^\infty_{\rm c}(M)$ that
$$
\big(J[H_0,A_0]J^*-[H,A]\big)\varphi=mL\varphi,
$$
where $L$ is a second order differential operator on $C^\infty_{\rm c}(M)$ with
coefficients in $C^\infty_{\rm b}(M)$ and support in $M_\infty$, and
$m\in C_{\rm b}(M)$ satisfies
$
\lim_{x\to\infty}
\big\|\big(m\circ\iota^{-1}\big)(x,\;\!\cdot\;\!)\big\|_{\linf(\Sigma)}=0
$.
It follows for all $\psi\in\langle H\rangle C^\infty_{\rm c}(M)$ that
$$
\langle H\rangle^{-1}\big(J[H_0,A_0]J^*-[H,A]\big)\langle H\rangle^{-1}\psi
=\langle H\rangle^{-1}mL\langle H\rangle^{-1}\psi.
$$
But we know from \cite[Lemma~1.6]{Sal01} that the operator $L\langle H\rangle^{-1}$,
defined on the dense set $\langle H\rangle C^\infty_{\rm c}(M)$, extends to an
element of $\B(\H)$. We also know from Lemma \ref{lemmecompact} that
$\langle H\rangle^{-1}m\equiv\big(m(H+i)^{-1}(H+i)\langle H\rangle^{-1}\big)^*$
belongs to $\K(\H)$. Therefore, there exists an operator $K\in\K(\H)$ such that
$\langle H\rangle^{-1}\big(J[H_0,A_0]J^*-[H,A]\big)\langle H\rangle^{-1}=K$ on $\H$,
which means that $J[H_0,A_0]J^*-[H,A]\in\K\big(\dom(H),\dom(H)^*\big)$.
\end{proof}

\begin{Lemma}\label{JJ*}
For each $\eta\in C^\infty_{\rm c}(\R)$, the operator $\eta(H)(JJ^*-1)\eta(H)$
belongs to $\K(\H)$.
\end{Lemma}

\begin{proof}
One has $JJ^*=\chi_\infty\;\!\iota^*\big(\j^2\otimes1\big)$, so $JJ^*-1$ acts as a
multiplication operator by a function in $C_{\rm c}^\infty(M)$. Therefore, the r.h.s.
of the equality
$$
\eta(H)(JJ^*-1)\eta(H)=\eta(H)(JJ^*-1)(H+i)^{-1}(H+i)\eta(H)
$$
is the product of one element of $\K(\H)$ and two elements of $\B(\H)$, due to Lemma
\ref{lemmecompact}.
\end{proof}

In the next statement, we use the notation $E^H(\lambda;\delta)$, with $\lambda\in\R$
and $\delta>0$, for the spectral projection
$E^H\big((\lambda-\delta,\lambda+\delta)\big)$.

\begin{Proposition}[Mourre estimate]\label{Mourre}
Suppose that Assumption \ref{CondsurgV} holds with $\muL>0$ and $\muS>1$. Then for
each $\lambda\in\R\setminus\T$, there exist $\delta,a>0$ and $K\in\K(\H)$ such that
$$
E^H(\lambda;\delta)[iH,A]E^H(\lambda;\delta)\ge a\;\!E^H(\lambda;\delta)+K.
$$
\end{Proposition}

\begin{proof}
The hypotheses (i), (ii), (iii) and (iv) of \cite[Thm.~3.1]{RT12} are verified in
Lemmas \ref{regul}, \ref{C2}, \ref{compacite1} and \ref{JJ*}, respectively. Moreover,
it is known (see for instance \cite[Sec.~3.1]{Tie06}) that $A_0$ is strictly
conjugate to $H_0$ on $\R\setminus\T$. So, the claim follows by applying
\cite[Thm.~3.1]{RT12}, keeping in mind that $A_0$ is conjugate to $H_0$ at
$\lambda\in\R$ if $A_0$ is strictly conjugate to $H_0$ at $\lambda$.
\end{proof}

\begin{Remark}[Critical values of $H$]
In the sequel, we call $\kappa(H):=\T\cup\sigma_{\rm p}(H)$ the set of critical
values of $H$. This terminology is motivated by the fact that Proposition
\ref{Mourre}, together with \cite[Thm.~7.2.13]{ABG}, implies that $A$ is strictly
conjugate to $H$ on $\R\setminus\kappa(H)$.
\end{Remark}

\subsection{Higher order resolvent estimates and absolute continuity}
\label{SecHighR}

The main result of this section is a statement on the differentiability of the
boundary values of the resolvent of $H$, which will be useful when discussing the
stationary formula for the scattering operator. Its proof is based on the abstract
approach developed in \cite{BG96}.

We start by introducing a multiplication operator $\Phi$ on $M$ given by
\begin{equation*}
\Phi\;\!\psi:=\chi_\infty\;\!\iota^*\big(\j^2\;\!\id\otimes1\big)\psi,
\quad\psi\in C_{\rm c}^\infty(M).
\end{equation*}
The closure of $\Phi$ in $\H$ (which we denote by the same symbol) is self-adjoint
\cite[Ex.~5.1.15]{Ped89} and equal to $J\Phi_0J^*$ on $C_{\rm c}^\infty(M)$.
Furthermore, for a map $h\in C\big(\R;\B(\H)\big)$ and any $s>0$ we say that $h$ is
Lipschitz-Zygmund continuous of class $\Lambda^s$ (in short $h\in\Lambda^s$) if
\begin{enumerate}
\item[(i)] $0<s<1$ and
$
\|h(x+\varepsilon)-h(x)\|_{\B(\H)}\le{\rm Const.}\;\!|\varepsilon|^s
$
for all $x\in \R$ and $|\varepsilon|\leq 1$,
\item[(ii)] $s=1$ and
$
\|h(x+\varepsilon)+h(x-\varepsilon)-2h(x)\|_{\B(\H)}\le{\rm Const.}\;\!|\varepsilon|
$ for all $x\in \R$ and $|\varepsilon|\leq 1$,
\item[(iii)] $s=k+\sigma$ with $k\in\N^*$ and $\sigma\in(0,1)$, and
$h\in C_{\rm b}^k(\R)$ with $k$-th derivative $h^{(k)}\in\Lambda^{\sigma}$.
\end{enumerate}
Now, we state the main result of this section.

\begin{Proposition}[Higher order resolvent estimates]\label{yeye}
Suppose that Assumption \ref{CondsurgV} holds with $\muL>0$ and $\muS>k$, for some
$k\in\N^*$. Take $\sigma\in(0,\min\{\muL,\muS-k,1\})$ and set $s:=k+\sigma-1/2$. Then for
$\lambda\in\R\setminus\kappa(H)$ and $\ell\in\{1,2,\ldots,k\}$, the limit
$
\langle A\rangle^{-s}(H-\lambda\mp i0)^{-\ell}\langle A\rangle^{-s}
:=\lim_{\varepsilon\searrow0}\langle A\rangle^{-s}(H-\lambda\mp i\varepsilon)^{-\ell}
\langle A\rangle^{-s}
$
exists in $\B(\H)$, and the map
\begin{equation}\label{horreur1}
\R\setminus\kappa(H)\ni\lambda\mapsto\langle A\rangle^{-s}(H-\lambda\mp i0)^{-1}
\langle A\rangle^{-s}\in\B(\H)
\end{equation}
is locally of class $\Lambda^{k-1+\sigma}$. In particular, the map \eqref{horreur1} is
$(k-1)$-times continuously differentiable, with derivative
\begin{equation}\label{horreur2}
\frac{\d^{k-1}}{\d \lambda^{k-1}}\;\!\langle A\rangle^{-s}(H-\lambda\mp i0)^{-1}
\langle A\rangle^{-s}=(k-1)!\;\!\langle A\rangle^{-s}(H-\lambda\mp i0)^{-k}
\langle A\rangle^{-s},
\end{equation}
and the map
$
\R\setminus\kappa(H)\ni\lambda\mapsto\langle A\rangle^{-s}(H-\lambda\mp i0)^{-k}
\langle A\rangle^{-s}\in\B(\H)
$
is locally of class $\Lambda^\sigma$.
\end{Proposition}

Before the proof, we recall that Lemma \ref{invariance} implies that the restriction
to $\G:=\dom(H)$ of the unitary group generated by $A$ defines a $C_0$-group in $\G$
as well as in its adjoint space $\G^*$ (\cf \cite[Prop.~3.2.5]{ABG}); we still denote
by $A$ the generators of these two $C_0$-groups. In particular, for any operator
$B\in\B(\G,\G^*)$, we write $B \in C^k(A;\G,\G^*)$ if the map
$\R\ni t\mapsto\e^{-itA}B\e^{itA}\in\B(\G,\G^*)$ is $k$-times strongly
differentiable. Similar definitions hold for the regularity classes $C^k(A;\G,\H)$
and $C^k(A;\H,\G)$.

\begin{proof}
(a) We prove the claim by applying \cite[Thm.~p.~12]{BG96} to our situation. So, we
only need to check the hypotheses of that theorem. For that purpose, we note that
$s>1/2$ and that $H$ has a spectral gap due to the lower bound $\Delta_M\ge0$ and the
boundedness of $V$. We also refer to point (b) below for a verification of the
hypothesis on the regularity of $H$ with respect to $A$. Thus,
\cite[Thm.~p.~12]{BG96} applies and the map \eqref{horreur1} is locally
Lipschitz-Zygmund of order $s-1/2$ on $\R\setminus\kappa(H)$. In particular, since
$s-1/2>k-1$, the map \eqref{horreur1} is $(k-1)$-times continuously differentiable
with bounded derivatives. The equality \eqref{horreur2} follows from the observation
made on pages $12$-$13$ of \cite{BG96}.

(b) For the regularity of $H$ with respect to $A$, it is necessary to show that $H$
is of class $\CC^{s+1/2}(A)\equiv\CC^{s+1/2,\infty}(A)$ (see \cite[Sec.~2.1]{BG96}).
By \cite[Prop.~5.2.2.(b)]{ABG}, we know that this holds if $H$ is of class $C^{k}(A)$
and if the $k$-iterated commutator $\ad^k_A\big((H-i)^{-1}\big)$ of $(H-i)^{-1}$ with
$A$ belongs to $\CC^\sigma(A)$ with $\sigma=s+1/2-k\in(0,1)$.

We first show that $H$ is of class $C^k(A)$. Since $\G$ is left invariant by the
group generated by $A$, and since $H$ is of class $C^1(A)$ with $[iH,A]\in\B(\G,\H)$
(see Lemma \ref{invariance} and its proof), Proposition 3.2 of \cite{RT12} tells us
it is enough to prove the inclusion $[H,A]\in C^{k-1}(A;\G,\H)$ (this condition
implies the weaker assumption $H \in C^{k-1}(A;\G,\H)\cap C^k(A;\G,\G^*)$). Let us
assume that $k>1$ since otherwise the proof is trivial. We know from
\cite[Thm.~5.1.3(b)]{ABG} that $D_1:=[H,A]\in C^1(A;\G,\H)$ if
\begin{equation}\label{ba_bound}
\liminf_{\tau\searrow0}\frac1\tau
\big\|\e^{i\tau A}D_1\e^{-i\tau A}-D_1\big\|_{\B(\G,\H)}<\infty.
\end{equation}
Now, a direct calculation using Assumption \ref{CondsurgV} with $\muL\ge0$ and
$\muS\ge k$ shows that there exists a second order differential operator $D_2$ with
coefficients in $C_{\rm b}^\infty(M)$ such that $[A,D_1]=D_2$ on
$C_{\rm c}^\infty(M)$. So, since
$\e^{it A}C_{\rm c}^\infty(M)\subset C_{\rm c}^\infty(M)$ for all $t\in\R$, one has
$$
\liminf_{\tau\searrow0}
\frac1\tau\big\|\e^{i\tau A}D_1\e^{-i\tau A}-D_1\big\|_{\B(\G,\H)}
=\liminf_{\tau\searrow0}\sup_{\psi\in C_{\rm c}^\infty(M),\,\|\psi\|_\G=1}
\left\|\int_0^1\d s\,\e^{i\tau sA}D_2\e^{-i\tau sA}\psi\right\|_\H,
$$
and one gets the bound \eqref{ba_bound} by noting that $\|D_2\|_{\B(\G,\H)}<\infty$
(due to \cite[Lemma~1.6]{Sal01}) and that
$\|\e^{itA}\psi\|_\G\le{\rm Const.}\;\!\|\psi\|_\G$ for all $t\in[0,1]$ (due to
\cite[Prop.~3.2.2.(b)]{ABG}). Thus $D_1 \in C^1(A;\G,\H)$, and this procedure can be
repeated iteratively (with $D_2$ replacing $D_1$, and so forth) to show that
$D_1\in C^{k-1}(A;\G,\H)$.

Let us now show that $\ad^k_A\big((H-i)^{-1}\big)$ belongs to $\CC^\sigma(A)$.
For that purpose, we first note that the inclusion $H\in C^k(A;\G,\H)$
implies by \cite[Prop.~5.1.6]{ABG} that $(H-i)^{-1}\in C^k(A;\H,\G)$.
Then, we observe that
\begin{align*}
\ad^k_A\big((H-i)^{-1}\big)
&=\ad^{k-1}_A\big(\big[(H-i)^{-1},A\big]\big)\\
&=-\ad^{k-1}_A\big((H-i)^{-1}[H,A](H-i)^{-1}\big)\\
&=\sum_{\substack{\ell_1,\ell_2,\ell_3\ge0\\\ell_1+\ell_2+\ell_3=k-1}}
{\sc c}_{\ell_1,\ell_2,\ell_3}\;\!\ad^{\ell_1}_A\big((H-i)^{-1}\big)\;\!
\ad^{\ell_2}_A\big([H,A]\big)\;\!\ad^{\ell_3}_A\big((H-i)^{-1}\big),
\end{align*}
with ${\sc c}_{\ell_1,\ell_2,\ell_3}\in\R$, $\ad^{\ell_1}_A\big((H-i)^{-1}\big)$ and
$\ad^{\ell_3}_A\big((H-i)^{-1}\big)$ in $C^1(A;\H,\G)\subset\CC^\sigma(A;\H,\G)$
and $\ad^{\ell_2}_A\big([H,A]\big)$ in $\B(\G,\H)$. Now, a duality argument implies
that $\ad^{\ell_1}_A\big((H-i)^{-1}\big)$ also belongs to
$\CC^\sigma(A;\G^*,\H)$. Thus, if one shows that $\ad^{\ell_2}_A\big([H,A]\big)$
belongs to $\CC^\sigma(A;\G,\G^*)$, then the statement would follow from an
application of \cite[Prop.~5.2.3.(a)]{ABG}. So, one is reduced to proving that
$D_{\ell_2}:=\ad^{\ell_2}_A\big([H,A]\big)\in\CC^\sigma(A;\G,\G^*)$ for any
$\ell_2\le k-1$, which is equivalent to
\begin{equation}\label{step2}
\big\|\e^{-itA}D_{\ell_2}\e^{itA}-D_{\ell_2}\big\|_{\B(\G,\G^*)}
\le{\rm Const.}\;\!|t|^\sigma\quad\hbox{for all }t\in(0,1).
\end{equation}
Now, algebraic manipulations as in \cite[p.~325]{ABG} together with the point (i) of
the proof of \cite[Prop.~7.5.7]{ABG} imply that
$$
\big\|\e^{-itA}D_{\ell_2}\e^{itA}-D_{\ell_2}\big\|_{\B(\G,\G^*)}
\le{\rm Const.}\;\!\big\|\sin(tA)D_{\ell_2}\big\|_{\B(\G,\G^*)}
\le{\rm Const.}\;\!\big\|tA(tA+i)^{-1}D_{\ell_2}\big\|_{\B(\G,\G^*)},
$$
with the constants independent of $t \in [0,1]$. Furthermore, if $A_t:=tA(tA+i)^{-1}$
and $\Lambda_t:=t\langle\Phi\rangle(t\langle\Phi\rangle+i)^{-1}$, then one has
$$
A_t=\big\{A_t+i(tA+i)^{-1}A\langle\Phi\rangle^{-1}\big\}\Lambda_t,
$$
with $A\langle\Phi\rangle^{-1}\in\B(\H,\G^*)$ due to \cite[Lemma~1.6]{Sal01}.
Finally, it is shown in the abstract framework of the proof of
\cite[Prop.~7.5.7]{ABG} that
$\|A_t\|_{B(\G^*)}+\|(tA+i)^{-1}\|_{\B(\G^*)}\le{\rm Const.}$ for all $t\in[0,1]$.
Thus, the estimate \eqref{step2} would hold if one shows that
$\|\Lambda_t D_{\ell_2}\|_{\B(\G,\H)}\le {\rm Const.}\;\!|t|^\sigma$.

For this, we recall that $D_{\ell_2}$ is (for any $\ell_2\le k-1$) equal on
$C_{\rm c}^\infty(M)$ to a second order differential operator with coefficients in
$C_{\rm b}^\infty(\R)$ if $\muL\ge0$ and $\muS\ge k$. But, since $\muL>0$, $\muS>k$
and $\sigma\le\min\{\muL,\muS-k\}$, the product $\langle\Phi\rangle^\sigma D_{\ell_2}$
is still a second second order
differential operator on $C_{\rm c}^\infty(M)$ with coefficients in
$C_{\rm b}^\infty(\R)$. It follows that
$$
\big\|\Lambda_tD_{\ell_2}\big\|_{\B(\G,\H)}
\le\big\|\Lambda_t\langle\Phi\rangle^{-\sigma}\big\|_{\B(\H)}\,
\big\|\langle\Phi\rangle^\sigma D_{\ell_2}\big\|_{\B(\G,\H)}
\le{\rm Const.}\;\!\sup_{x\in\R}\big|t\langle x\rangle^{1-\sigma}
(t\langle x\rangle+i)^{-1}\big|
\le{\rm Const.}\;\!|t|^\sigma,
$$
as required.
\end{proof}

The nature of the spectrum of $H$ can now be deduced:

\begin{Proposition}[Spectral properties of $H$]\label{natureH}
Suppose that Assumption \ref{CondsurgV} holds with $\muL>0$ and $\muS>1$. Then, the
point spectrum of $H$ in $\R\setminus\T$ is composed of eigenvalues of finite
multiplicity with no accumulation point. Furthermore, the operator $H$ has no
singular continuous spectrum.
\end{Proposition}

\begin{proof}
We know from Proposition \ref{Mourre} that a Mourre estimate holds for $H$. We also
know from the proof of Proposition \ref{yeye} with $k=1$ that the operator $H$ is of
class $\CC^{1+\sigma}(A)$ for any $\sigma\in(0,\min\{\muL,\muS-1,1\})$. So, $H$ is a
fortiori of class $C^{1,1}(A)$. Finally, we recall that $H$ has a spectral gap, as
mentioned in the proof of Proposition \ref{yeye}. Therefore, one can simply apply
\cite[Thm.~7.4.2]{ABG} to obtain the stated results (note that \cite[Thm.~7.4.2]{ABG}
only implies that $H$ has no singular continuous spectrum in $\R\setminus\T$, but
since $\T$ is countable this implies that $H$ has no singular continuous spectrum at
all).
\end{proof}

\subsection{From one weight to another}

The higher order resolvent estimates for $H$ obtained in Proposition \ref{yeye} are
formulated in terms of the weights $\langle A\rangle^{-s}$. In applications, such as
the mapping properties of the scattering operator, it is often more convenient to
deal with weights defined in terms of multiplication operators. So, we devote this
subsection to the derivation of higher order resolvent estimates for $H$ in terms of
the weights $\langle\Phi\rangle^{-s}$.

We start by recalling a similar result for the pair $(H_0,\Phi_0)$ that can be
deduced from the proof of \cite[Lemma~3.6]{Tie06}:

\begin{Lemma}\label{passage}
Let $s>k-1/2$ for some $k\in\N^*$. Then for $\lambda\in\R\setminus\T$ and
$\ell\in\{1,2,\ldots,k\}$, the limit
$
\langle\Phi_0\rangle^{-s}(H_0-\lambda\mp i0)^{-\ell}\langle\Phi_0\rangle^{-s}
:=\lim_{\varepsilon\searrow0}\langle\Phi_0\rangle^{-s}
(H_0-\lambda\mp i\varepsilon)^{-\ell}\langle\Phi_0\rangle^{-s}
$
exists in $\B(\H_0)$, and the map
$$
\R\setminus\T\ni\lambda\mapsto\langle\Phi_0\rangle^{-s}(H_0-\lambda\mp i0)^{-1}
\langle\Phi_0\rangle^{-s}\in\B(\H_0)
$$
is $(k-1)$-times continuously differentiable, with derivative
$$
\frac{\d^{k-1}}{\d\lambda^{k-1}}\;\!\langle\Phi_0\rangle^{-s}(H_0-\lambda\mp i0)^{-1}
\langle\Phi_0\rangle^{-s}=(k-1)!\;\!\langle\Phi_0\rangle^{-s}(H_0-\lambda\mp i0)^{-k}
\langle\Phi_0\rangle^{-s}.
$$
\end{Lemma}

We turn now to the derivation of similar resolvent estimates for $H$ in terms of the
weights $\langle\Phi\rangle^{-s}$.

\begin{Proposition}\label{yep}
Suppose that Assumption \ref{CondsurgV} holds with $\muL>0$ and $\muS>k$, for some
$k\in\N^*$. Take $\sigma\in(0,\min\{\muL,\muS-k,1\})$ and set $s:=k+\sigma-1/2$. Then for
$\lambda\in\R\setminus\kappa(H)$ and $\ell\in\{1,2,\ldots,k\}$, the limit
$
\langle\Phi\rangle^{-s}(H-\lambda\mp i0)^{-\ell}\langle\Phi\rangle^{-s}
:=\lim_{\varepsilon\searrow0}\langle\Phi\rangle^{-s}
(H-\lambda\mp i\varepsilon)^{-\ell}\langle\Phi\rangle^{-s}
$
exists in $\B(\H)$, and the map
\begin{equation}\label{map_map}
\R\setminus\kappa(H)\ni\lambda\mapsto\langle\Phi\rangle^{-s}(H-\lambda\mp i0)^{-1}
\langle\Phi\rangle^{-s}\in\B(\H)
\end{equation}
is locally of class $\Lambda^{k-1+\sigma}$.
In particular, the map \eqref{map_map} is $(k-1)$-times continuously differentiable
with $(k-1)$-th derivative locally of class $\Lambda^\sigma$.
\end{Proposition}

\begin{proof}
Take $z\in\C\setminus\sigma(H)$, fix $\lambda_0\in\R\setminus\sigma(H)$ and let
$m\in\N^*$ with $2m\ge s$. Then, by applying iteratively $m$-times the formula
\cite[Eq.~7.4.2]{ABG} for the resolvent $R(z):=(H-z)^{-1}$ one obtains that
$$
R(z)=(z-\lambda_0)^{2m}R(\lambda_0)^mR(z)R(\lambda_0)^m+I(z,\lambda_0,m),
$$
where $I(z,\lambda_0,m)$ is a polynomial in $z$ with coefficients in $\B(\H)$. It
follows that
\begin{align}
\langle\Phi\rangle^{-s}R(z)\langle\Phi\rangle^{-s}
&=(z-\lambda_0)^{2m}\langle\Phi\rangle^{-s}R(\lambda_0)^mR(z)R(\lambda_0)^m
\langle\Phi\rangle^{-s}
+\langle\Phi\rangle^{-s}I(z,\lambda_0,m)\langle\Phi\rangle^{-s}\nonumber\\
&=(z-\lambda_0)^{2m}\big\{\langle\Phi\rangle^{-s}R(\lambda_0)^m
\langle A\rangle^s\big\}\langle A\rangle^{-s}R(z)\langle A\rangle^{-s}
\big\{\langle A\rangle^sR(\lambda_0)^m\langle\Phi\rangle^{-s}\big\}\nonumber\\
&\qquad+\langle\Phi\rangle^{-s}I(z,\lambda_0,m)\langle\Phi\rangle^{-s}.\label{bocata}
\end{align}
Furthermore, it is proved in Lemma \ref{estimation} that
$B:=\langle A\rangle^sR(\lambda_0)^m\langle\Phi\rangle^{-s}$ belongs to $\B(\H)$. So,
\eqref{bocata} can be written as
$$
\langle\Phi\rangle^{-s}R(z)\langle\Phi\rangle^{-s}
=(z-\lambda_0)^{2m}B^*\langle A\rangle^{-s}R(z)\langle A\rangle^{-s}B
+\langle\Phi\rangle^{-s}I(z,\lambda_0,m)\langle\Phi\rangle^{-s}.
$$
This last identity (with $z=\lambda\pm i\varepsilon$), together with Proposition
\ref{yeye}, implies the claim.
\end{proof}

\section{Scattering theory}\label{Sec_Scatt}
\setcounter{equation}{0}

In this section, we present the standard short-range scattering theory for our model.
Accordingly, we formulate all our statements in terms of the common exponent
$\mu\equiv\min\{\muL,\muS\}$ to ensure that both the short-range and long-range
perturbations decay at least as $\langle x\rangle^{-\mu}$ at infinity. As usual, the
assumption $\mu>1$ is sufficient to guarantee the existence and the asymptotic
completeness of the wave operators.

\subsection{Existence of the wave operators}

This first subsection deals with the existence of the wave operators and some of
their properties. Mourre theory as developed in the previous section is not necessary
for that part of the investigations. However, once the problem of the asymptotic
completeness will be addressed, all the results obtained so far will be necessary.

We start with two lemmas which will play a key role when proving the existence of the
wave operators. Their statement involves the sets $\D_t\subset\dom(H_0)$, $t\ge0$,
defined by
\begin{equation*}
\D_t:=\big\{\varphi\in\H_t(\R)\otimes\ltwo(\Sigma)\mid\varphi=\eta(H_0)\varphi
\hbox{ for some }\eta\in C^\infty_{\rm c}(\R\setminus\T)\big\}.
\end{equation*}

\begin{Lemma}\label{majT}
Suppose that Assumption \ref{CondsurgV} holds with $\mu>0$, let $\varphi\in\D_s$ with
$s>0$ and take $\mu'<\min\{\mu,s\}$. Then, one has for any $t\in\R$
\begin{equation}\label{h0}
\big\|T\e^{-itH_0}\varphi\big\|_\H\le{\rm Const.}\;\!(1+|t|)^{-\mu'}.
\end{equation}
\end{Lemma}

\begin{proof}
Let $\varphi\in\D_s$. Then, one deduces from Lemma \ref{preuvepourRafa} that
$$
\big\|T\e^{-itH_0}\varphi\big\|_\H\le{\rm Const.}\;\!
\big\|\langle H_0\rangle\langle\Phi_0\rangle^{-\mu}\e^{-itH_0}\varphi\big\|_{\H_0}.
$$
Furthermore, since $\D_s\subset\D_0$ there exist
$\eta\in C_{\rm c}^\infty(\R\setminus\T)$ and $j_0\in\N$ such that
$$
\e^{-itH_0}\varphi
=\sum_{j=0}^{j_0}\big(\e^{-it(P^2+\tau_j)}\otimes \P_j\big)\;\!\eta(H_0)\varphi
=\sum_{j=0}^{j_0}\left(\e^{-it(P^2+\tau_j)}\eta(P^2+\tau_j)\otimes\P_j\right)\varphi.
$$
As a consequence, one obtains that
$$
\big\|T\e^{-itH_0}\varphi\big\|_\H
\le{\rm Const.}\;\!\sum_{j=0}^{j_0}\big\|\big\langle P^2+\tau_j\big\rangle
\langle Q\rangle^{-\mu}\e^{-it(P^2+\tau_j)}\eta_j(P^2)\langle Q\rangle^{-s}
\big\|_{\B(\ltwo(\R))}\;\!,
$$
with $\eta_j:=\eta(\;\!\cdot\;\!+\tau_j)$. Then, some commutators calculations lead
to the estimate
\begin{align}
\big\|T\e^{-itH_0}\varphi\big\|_\H
&\le{\rm Const.}\;\!\sum_{j=0}^{j_0}\big\|\langle Q \rangle^{-\mu}\e^{-itP^2}P^2
\eta_j(P^2)\langle Q\rangle^{-s}\big\|_{\B(\ltwo(\R))}\label{h1}\\
&\quad+{\rm Const.}\;\!\sum_{j=0}^{j_0}\big\|\langle Q \rangle^{-\mu}\e^{-itP^2}P
\eta_j(P^2)\langle Q\rangle^{-s}\big\|_{\B(\ltwo(\R))}\label{h2}\\
&\quad+{\rm Const.}\;\!\sum_{j=0}^{j_0}\big\|\langle Q \rangle^{-\mu} \e^{-itP^2}
\eta_j(P^2)\langle Q\rangle^{-s}\big\|_{\B(\ltwo(\R))}\label{h3}.
\end{align}
Since $0\notin\supp(\eta_j)$, one can apply \cite[Lemma~9]{ACS87} to infer that
\eqref{h1} and \eqref{h3} are bounded by the r.h.s. of \eqref{h0} with
$\mu'<\min\{\mu,s\}$. For \eqref{h2}, one first uses the equality
$$
P\;\!\eta_j(P^2)\langle Q\rangle^{-s}
=\big\{\langle P\rangle\;\!\eta_j(P^2)\langle Q\rangle^{-s}\big\}
\big\{\langle Q\rangle^sP\langle P\rangle^{-1}\langle Q\rangle^{-s}\big\},
$$
and then the same bound can be obtained by taking  \cite[Lemma~9]{ACS87} and
\cite[Lemma~1]{ACS87} into account.
\end{proof}

For the next lemma, we introduce the subspaces $\H_0^\pm$ of $\H_0$ given by
\begin{equation}\label{def_H+-}
\H_0^\pm
:=\big\{\varphi\in\H_0\mid\supp(\F\otimes1)\varphi\subset\R_\pm\times\Sigma\big\},
\end{equation}
where $\R_+:=(0,\infty)$ and $\R_-:=(-\infty,0)$.

\begin{Lemma}\label{deAJ}
Let $s>0$ and $\varphi_\pm\in\D_s\cap\H_0^{\pm}$. Then, one has
$$
\big\|(J^*J-1)\e^{-itH_0}\varphi_\pm\big\|_{\H_0}\le{\rm Const.}\;\!(1+|t|)^{-s}
\quad\hbox{for any }\,t\in\R_\pm.
$$
\end{Lemma}

\begin{proof}
The proof of this statement relies on estimates obtained in \cite[Sec.~II.A]{AJ07} in
the context of $1$-dimensional anisotropic scattering. In \cite[Eq.~17]{AJ07} it is
proved that if $\psi\in\H_s(\R)$ with $\supp(\F\psi)\subset\R_+$, then one has for
each $x_0\in\R$ and $t>0$
$$
\big\|\chi_{(-\infty,x_0)}\e^{-itP^2}\psi\big\|_{\ltwo(\R)}
\le{\rm Const.}\;\!(1+|t|)^{-s}.
$$
A similar estimate with $t<0$ also holds if $\psi\in\H_s(\R)$ and
$\supp(\F\psi)\subset\R_-$ (see \cite[Eq.~20]{AJ07}).

Now, it is easily observed that
$$
J^*J-1
=\big(\j^2-1\big)\otimes1
=\chi_{(-\infty,2)}\big(\j^2-1\big)\otimes1.
$$
So, one obtains
$$
\big\|(J^*J-1)\e^{-itH_0}\varphi\big\|_{\H_0}
\le\big\|\big(\chi_{(-\infty,2)}\e^{-itP^2}\otimes1\big)\varphi\big\|_{\H_0}
$$
for any $t\in\R$ and $\varphi\in\H_0$. This, together with the extensions of the
mentioned estimates to the algebraic tensor product $\ltwo(\R)\odot\ltwo(\Sigma)$,
implies the claim for vectors $\varphi_\pm\in\D_s\cap\H_0^{\pm}$ and $t\in\R_\pm$.
\end{proof}

\begin{Proposition}[Existence of the wave operators]\label{existence}
Suppose that Assumption \ref{CondsurgV} holds with $\mu>1$. Then, the generalized
wave operators
$$
W_\pm:=\slim_{t\to\pm\infty}\e^{itH}J\e^{-itH_0}
$$
exist and are partial isometries with initial subspaces $\H_0^\pm$.
\end{Proposition}

\begin{proof}
The existence of the wave operators is based on the Cook-Kuroda method. One first
observes that, since $J\in\B\big(\dom(H_0),\dom(H)\big)$, the following equality
holds for any $\varphi\in\dom(H_0):$
$$
\e^{itH}J\e^{-itH_0}\varphi=J\varphi+i\int_0^t\d s\,\e^{isH}T\e^{-isH_0}\varphi.
$$
Furthermore, if $\varphi\in\D_\mu \subset\dom(H_0)$ it follows from Lemma \ref{majT}
that there exists $\mu'\in(1,\mu)$ such that
$$
\int_{-\infty}^\infty\d s\,\big\|\e^{isH}T\e^{-isH_0}\varphi\big\|_\H
\le{\rm Const.}\int_{-\infty}^\infty\d s\,(1+|s|)^{-\mu'}<\infty.
$$
Since $\D_\mu$ is dense in $\H_0$, this estimate implies the existence of both wave
operators $W_\pm$.

We now show that $W_\pm\H_0^\mp=\{0\}$. Assume that $\varphi_\pm\in\D_s\cap\H_0^\pm$
for some $s>0$. Then, one has
\begin{align*}
\|W_\pm\varphi_\mp\|_\H
=\lim_{t\to\pm\infty}\big\|\e^{itH}J\e^{-itH_0}\varphi_\mp\big\|_\H
&=\lim_{t\to\pm\infty}\big\|J\chi_{(0,\infty)}\e^{-itH_0}\varphi_\mp\big\|_\H\\
&\le{\rm Const.}\;\!\lim_{t\to\pm\infty}
\big\|\big(\chi_{(0,\infty)}\e^{-itP^2}\otimes1\big)\varphi_\mp\big\|_{\H_0}\\
&\le{\rm Const.}\;\!\lim_{t\to\pm\infty}(1+|t|)^{-s}\\
&=0,
\end{align*}
where we have used for the last inequality the extension of the estimates
\cite[Eq.~18~\&~19]{AJ07} to the algebraic tensor product
$\ltwo(\R)\odot\ltwo(\Sigma)$. Since $\D_s\cap\H_0^\mp$ is dense in $\H_0^\mp$, one
infers that $W_\pm\H_0^\mp=\{0\}$.

Finally, we show that $\|W_\pm\varphi_\pm\|_\H=\|\varphi_\pm\|_{\H_0}$ for each
$\varphi_\pm\in\H_0^\pm$. One has for any $\varphi_\pm\in\D_s\cap\H_0^\pm$
\begin{align*}
\big|\|W_\pm\varphi_\pm\|^2_\H-\|\varphi_\pm\|^2_{\H_0}\big|
&=\lim_{t\to\pm\infty}
\big|\big\langle\e^{-itH_0}\varphi_\pm,(J^*J-1)\e^{-itH_0}\varphi_\pm\big\rangle_{\H_0}\big|\\
&\le{\rm Const.}\;\!\lim_{t\to\pm\infty}
\big\|(J^*J-1)\e^{-itH_0}\varphi_\pm\big\|_{\H_0}\\
&=0,
\end{align*}
due to Lemma \ref{deAJ}. So, the statement follows by the density of
$\D_s\cap\H_0^{\pm}$ in $\H_0^\pm$.
\end{proof}

Finally, we present an estimate which is going to play an important role when proving
the existence of the time delay. Its proof relies on estimates obtained so far in
this section.

\begin{Lemma}\label{condL1}
Suppose that Assumption \ref{CondsurgV} holds with $\mu>2$. Then, one has for any
$\varphi_\pm\in\D_\mu\cap\H_0^\pm$
$$
\big\|\big(J^*W_\pm-1\big)\e^{-itH_0}\varphi_\pm\big\|_{\H_0}\in\lone(\R_\pm,\d t).
$$
\end{Lemma}

\begin{proof}
Let $\varphi\in\D_\mu$. Then, we know from Lemma \ref{majT} that there exists
$\mu'\in(2,\mu)$ such that
$$
\big\|J^*(W_--J)\e^{-itH_0}\varphi\big\|_{\H_0}
\le{\rm Const.}\;\!\int_{-\infty}^t\d s\,\big\|T\e^{-isH_0}\varphi\big\|_\H
\le{\rm Const.}\;\!\int_{-\infty}^t\d s\,(1+|s|)^{-\mu'}
\in\lone(\R_-,\d t).
$$
A similar argument shows that $\|J^*(W_+-J)\e^{-itH_0}\varphi\|_{\H_0}$ belongs to
$\lone(\R_+,\d t)$. Furthermore, one obtains from Lemma
\ref{deAJ} that
$$
\big\|(J^*J-1)\e^{-itH_0}\varphi_\pm\big\|_{\H_0}\in\lone(\R_\pm,\d t).
$$
for each $\varphi_\pm\in\D_\mu\cap\H_0^\pm$. Since
$J^*W_\pm-1=J^*(W_\pm-J)+(J^*J-1)$, the combination of both estimates implies the
claim.
\end{proof}

\subsection{Asymptotic completeness of the wave operators}

We establish in this subsection the asymptotic completeness of the wave operators
$W_\pm$ by applying the abstract criterion \cite[Prop.~5.1]{RT12}. To do so, we need
two preliminary lemmas.

\begin{Lemma}\label{small}
One has $\slim_{t\to\pm\infty}(JJ^*-1)\e^{-itH}P_{\rm ac}(H)=0$.
\end{Lemma}

\begin{proof}
We know from the proof of Lemma \ref{JJ*} that $(JJ^*-1)(H+i)^{-1}\in\K(\H)$. So, one
can conclude using a classical propagation estimate for vectors in $P_{\rm ac}(H)\H$
(see \cite[Prop.~5.7.(b)]{Amr09}).
\end{proof}

\begin{Lemma}\label{existence_le_retour}
Suppose that Assumption \ref{CondsurgV} holds with $\mu>1$. Then, the following
wave operators exist:
$$
W_\pm(H_0,H,J^*):=\slim_{t\to\pm\infty}\e^{itH_0}J^*\e^{-itH}P_{\rm ac}(H).
$$
\end{Lemma}

\begin{proof}
We follow the standard approach (see \eg \cite[Cor.~4.5.7]{Yaf92}) by showing that
$HJ-JH_0$ admits for all $\varphi\in\dom(H_0)$ and $\psi\in\dom(H)$ a (sesquilinear
form) decomposition
\begin{equation}\label{flifli}
\langle J\varphi,H\psi\rangle_\H-\langle JH_0\varphi,\psi\rangle_\H
=\langle G_0\varphi,G\psi\rangle_\H,
\end{equation}
where $G_0:\H_0\to\H$ is $H_0$-bounded and locally $H_0$-smooth on $\R\setminus\T$
and $G:\H\to\H$ is $H$-bounded and locally $H$-smooth on $\R\setminus\kappa(H)$ (with
$\kappa(H)$ being of measure zero).

For that purpose, one first fixes $s\in(1/2,\mu-1/2)$ and shows as in Lemma
\ref{preuvepourRafa} that the operator
$\langle\Phi\rangle^sT\langle\Phi_0\rangle^{\mu-s}$ defined on
$\S(\R)\odot C^\infty(\Sigma)$ extends continuously to an operator
$B_s\in\B\big(\dom(H_0),\H\big)$. It follows by Proposition \ref{mapping_H_not}.(i)
that there exists an operator $C_s\in\B(\H_0,\H)$ such that one has on
$\S(\R)\odot C^\infty(\Sigma)$
$$
\langle \Phi\rangle ^sT
=B_s(H_0-i)^{-1}(H_0-i)\langle \Phi_0\rangle^{s-\mu}
=C_s\langle \Phi_0\rangle^{s-\mu}(H_0-i).
$$
Thus, one gets for any $\varphi\in\S(\R)\odot C^\infty(\Sigma)$ and $\psi\in\dom(H)$
the equalities
\begin{equation}\label{fluflu}
\langle J\varphi,H\psi\rangle_\H-\langle JH_0\varphi,\psi\rangle_\H
=\big\langle\langle\Phi\rangle^sT\varphi,\langle\Phi\rangle^{-s}\psi\big\rangle_\H\\
=\big\langle C_s\langle\Phi_0\rangle^{s-\mu}(H_0-i)\varphi,
\langle\Phi\rangle^{-s}\psi\big\rangle_\H,
\end{equation}
which extend to all $\varphi\in\dom(H_0)$ due to the density of
$\S(\R)\odot C^\infty(\Sigma)$ in $\dom(H_0)$.

Now, the operator $G(s):=\langle\Phi\rangle^{-s}$ is $H$-bounded and locally
$H$-smooth on $\R\setminus\kappa(H)$ due to Proposition \ref{yep} with $k=1$, and the
operator $G_0(s):=C_s\langle \Phi_0\rangle^{s-\mu}(H_0-i)$ is $H_0$-bounded and
$H_0$-smooth on $\R\setminus\T$ due to a simple calculation. So, the decomposition
\eqref{fluflu} is equivalent to \eqref{flifli}, and the claim is proved.
\end{proof}

We are finally in a position to prove the asymptotic completeness of the wave
operators:

\begin{Proposition}[Asymptotic completeness of the wave operators]\label{asymptotc}
Suppose that Assumption \ref{CondsurgV} holds with $\mu>1$. Then,
$\Ran\big(W_\pm(H,H_0,J)\big)=\H_{\rm ac}(H)$.
\end{Proposition}

\begin{proof}
This result follows from \cite[Prop.~5.1]{RT12}, whose assumptions have been checked
in Proposition \ref{existence}, Lemma \ref{small} and Lemma
\ref{existence_le_retour}.
\end{proof}

\begin{Remark}
Let us collect some information about the spectrum of the operator $H$. Under the Assumption \ref{CondsurgV} with $\muL>0$ and $\muS>1$, a Mourre estimate was obtained in Proposition \ref{Mourre} based on the abstract scheme presented in \cite[Thm.~3.1]{RT12}. It also  follows from this abstract result, together with \cite[Prop.~7.2.6]{ABG}, that $\sigma_{\rm ess}(H)\subset \sigma_{\rm ess}(H_0)=[0,\infty)$,
and thus
$\sigma_{\rm ac}(H)\subset[0,\infty)$, since
$\sigma_{\rm sc}(H)=\varnothing$ due to Proposition \ref{natureH}.

More can be said under (the stronger) Assumption \ref{CondsurgV} with $\mu>1$:
We know by Lemma \ref{Lem_F_0}.(c) that the restrictions
$H_0^\pm:=H_0\upharpoonright\H_0^\pm$ are self-adjoint operators with spectrum
$\sigma\big(H_0^\pm\big)=\sigma_{\rm ac}\big(H_0^\pm\big)=[0,\infty)$. We also know
by Propositions \ref{existence} and \ref{asymptotc} that the absolutely continuous
parts of $H_0^\pm$ and $H$ are unitarily equivalent. So, one has
$\sigma_{\rm ac}(H)=[0,\infty)$, and we deduce from Proposition \ref{natureH} that
$$
\sigma_{\rm ess}(H)=\sigma_{\rm ac}(H)=[0,\infty).
$$
\end{Remark}

\subsection{Stationary formula for the scattering operator}\label{Sec_Sta-Sta}

In simple situations, the scattering operator is defined as the product $W_+^*W_-$
from $\H_0$ to $\H_0$. However, in our setup, this product is not unitary since the
wave operators are partial isometries with nontrivial kernels. Therefore, we define
instead the scattering operator as
$$
S:=W_+^*W_-:\H_0^-\to\H_0^+,
$$
and note that this operator is unitary due to the asymptotic completeness established
in Proposition \ref{asymptotc} (see \eqref{def_H+-} for the definition of the spaces
$\H_0^\pm\subset \H_0$). Since the scattering operator $S$ commutes with the free
evolution group $\{\e^{itH_0}\}_{t\in\R}$, one infers from Lemma \ref{Lem_F_0}.(c)
that $S$ admits a direct integral decomposition
\begin{equation*}
F_0SF_0^{-1}
=\int_{[0,\infty)}^\oplus\d\lambda\,S(\lambda):F_0\H_0^-\to F_0\H_0^+,
\end{equation*}
where $S(\lambda)$ (the scattering matrix at energy $\lambda$) is an operator acting
unitarily from $\H_0^-(\lambda):=\big(F_0\H_0^-\big)(\lambda)$ to
$\H_0^+(\lambda):=\big(F_0\H_0^+\big)(\lambda)$. Here, the subspaces
$\H_0^\pm(\lambda)\subset\H_0(\lambda)$ satisfy
$$
\H_0^-(\lambda)= \bigoplus_{j\in\N(\lambda)}\P_j\;\!\ltwo(\Sigma)\oplus\{0\}
\qquad\hbox{and}\qquad
\H_0^+(\lambda)= \bigoplus_{j\in\N(\lambda)}\{0\}\oplus\P_j\;\!\ltwo(\Sigma),
$$
and are embedded in
$
\H_0^-(\infty):=\bigoplus_{j\in\N}\P_j\;\!\ltwo(\Sigma)\oplus\{0\}
$
and
$
\H_0^+(\infty):=\bigoplus_{j\in\N}\{0\}\oplus\P_j\;\!\ltwo(\Sigma)
$.

In the sequel, we derive a formula for the operators $S(\lambda)$ by using the
stationary scattering theory of \cite[Sec.~5.5]{Yaf92}. Our first step toward that
formula is the following lemma; recall that $\T\equiv\{\tau_j\}_{j\in\N}$ is the
spectrum of $\Delta_\Sigma$ in $\ltwo(\Sigma)$ and that
$G_0(s)\in\B\big(\dom(H_0),\H\big)$, with $s\in(1/2,\mu-1/2)$, was defined in the
proof of Lemma \ref{existence_le_retour}.

\begin{Lemma}\label{lemme_Z_0}
Suppose that Assumption \ref{CondsurgV} holds with $\mu>1$ and let
$\lambda\in[0,\infty)\setminus\T$. Then,
\begin{enumerate}
\item[(a)] for any $s\in(1/2,\mu-1/2)$, the operator
$Z_0\big(\lambda;G_0(s)\big):\H\to\H_0(\lambda)$ given by
$$
Z_0\big(\lambda;G_0(s)\big)\psi:=\big(F_0G_0(s)^*\psi\big)(\lambda),
\quad\psi\in\dom(H),
$$
is well-defined and extends to an element of $\B\big(\H,\H_0(\lambda)\big)$ which we
denote by the same symbol,
\item[(b)] if $\mu>k+1$ for some $k\in\N$, and if $s\in(1/2,\mu-k-1/2)$, the function
$
[0,\infty)\setminus\T\ni\lambda\mapsto Z_0\big(\lambda;G_0(s)\big)
\in\B\big(\H,\H_0(\infty)\big)
$
is locally $k$-times H\"older continuously differentiable,
\item[(c)] for all $s_1,s_2\in(1/2,\mu-1/2)$, one has
$$
Z_0\big(\lambda;G_0(s_1)\big)\langle\Phi\rangle^{-s_1}
=Z_0\big(\lambda;G_0(s_2)\big)\langle\Phi\rangle^{-s_2}.
$$
\end{enumerate}
\end{Lemma}

\begin{proof}
The three claims are proved, respectively, in points (a), (b) and (c) below. In the
proofs, we freely use the following inclusions which can be established as in Lemma
\ref{preuvepourRafa}: Given $s_1,s_2\in\R$ with $s_1+s_2\le\mu$, one has
$$
L(s_1,s_2):=\overline{(H-i)^{-1}\langle\Phi\rangle^{s_1}T\langle\Phi_0\rangle^{s_2}
\upharpoonright\S(\R)\odot C^\infty(\Sigma)}\in\B(\H_0,\H)
$$
and
$$
R(s_1,s_2):=\overline{\langle\Phi\rangle^{s_1}T\langle\Phi_0\rangle^{s_2}(H_0+i)^{-1}
\upharpoonright\S(\R)\odot C^\infty(\Sigma)}\in\B(\H_0,\H).
$$

(a) Take $\psi\in\dom(H)$, $\varphi\in\dom(H_0)$ and
$\{\varphi_n\}_{n\in\N}\subset\S(\R)\odot C^\infty(\Sigma)$ such that
$\lim_{n\to\infty}\|\varphi-\varphi_n\|_{\dom(H_0)}=0$. Then, we have for any fixed
$s\in(1/2,\mu-1/2)$
\begin{align*}
\big\langle\psi,G_0(s)\varphi\big\rangle_\H
=\lim_{n\to\infty}\big\langle\psi,\langle\Phi\rangle^sT\varphi_n\big\rangle_\H
&=\lim_{n\to\infty}\big\langle(H+i)\psi,(H-i)^{-1}\langle\Phi\rangle^sT
\langle\Phi_0\rangle^{\mu-s}\langle\Phi_0\rangle^{s-\mu}\varphi_n\big\rangle_\H\\
&=\lim_{n\to\infty}\big\langle(H+i)\psi,L(s,\mu-s)\langle\Phi_0\rangle^{s-\mu}
\varphi_n\big\rangle_{\H}\\
&=\big\langle\langle\Phi_0\rangle^{s-\mu}L(s,\mu-s)^*(H+i)\psi,
\varphi\big\rangle_{\H_0},
\end{align*}
meaning that $G_0(s)^*\psi=\langle\Phi_0\rangle^{s-\mu}L(s,\mu-s)^*(H+i)\psi$. It
follows by Lemma \ref{Lem_F_0}.(a) that for each $\lambda\in[0,\infty)\setminus\T$
\begin{align}
&Z_0\big(\lambda;G_0(s)\big)\psi\nonumber\\
&=F_0(\lambda)\langle\Phi_0\rangle^{s-\mu}L(s,\mu-s)^*(H+i)\psi\label{Form_L}\\
&=\textstyle F_0(\lambda)\big\{1\otimes\big(\sum_{j\in\N(\lambda)}\P_j\big)\big\}
\langle\Phi_0\rangle^{s-\mu}L(s,\mu-s)^*(H+i)\psi\nonumber\\
&=\textstyle F_0(\lambda)\big(\langle Q\rangle^{s-\mu}\otimes1\big)
\big\{1\otimes\big(\sum_{j\in\N(\lambda)}\P_j\big)\big\}
(H_0-i)(H_0-i)^{-1}L(s,\mu-s)^*(H+i)\psi\nonumber\\
&=F_0(\lambda)\big(\langle Q\rangle^{s-\mu}\otimes1\big)
\big\{(P^2-i)\otimes1+1\otimes(\Delta_\Sigma)_\lambda\big\}
(H_0-i)^{-1}L(s,\mu-s)^*(H+i)\psi,\nonumber
\end{align}
where
$
(\Delta_\Sigma)_\lambda
:=\sum_{j\in\N(\lambda)}\tau_j\;\!\P_j\in\B\big(\ltwo(\Sigma)\big)
$.
Now, a direct calculation using Proposition \ref{mapping_H_not}.(iii) shows for all
$\widetilde\varphi\in\S(\R)\odot C^\infty(\Sigma)$ that
$$
\big\langle\widetilde\varphi,(H_0-i)^{-1}L(s,\mu-s)^*(H+i)\psi\big\rangle_{\H_0}
=\big\langle\widetilde\varphi,R(s,\mu-s)^*\psi\big\rangle_{\H_0}.
$$
So, one infers that $(H_0-i)^{-1}L(s,\mu-s)^*(H+i)\psi=R(s,\mu-s)^*\psi$ by the
density of $\S(\R)\odot C^\infty(\Sigma)$ in $\H_0$, and thus
$$
Z_0\big(\lambda;G_0(s)\big)\psi
=F_0(\lambda)\big(\langle Q\rangle^{s-\mu}\otimes1\big)
\big\{(P^2-i)\otimes1+1\otimes(\Delta_\Sigma)_\lambda\big\}R(s,\mu-s)^*\psi.
$$
Now, the operator on the r.h.s. belongs to $\B\big(\H,\H_0(\lambda)\big)$ due to
Lemma \ref{Lem_F_0}.(a). So, one obtains that
\begin{equation}\label{form_Z_0}
\overline{Z_0\big(\lambda;G_0(s)\big)\upharpoonright\dom(H)}
=F_0(\lambda)\big(\langle Q\rangle^{s-\mu}\otimes1\big)
\big\{(P^2-i)\otimes1+1\otimes(\Delta_\Sigma)_\lambda\big\}R(s,\mu-s)^*,
\end{equation}
which proves the first claim.

(b) Write $Z_0\big(\lambda;G_0(s)\big)$ for the closure
$\overline{Z_0\big(\lambda;G_0(s)\big)\upharpoonright\dom(H)}$ and fix an interval
$(\tau_j,\tau_{j+1})$. Then, the function
$$
(\tau_j,\tau_{j+1})\ni\lambda\mapsto Z_0\big(\lambda;G_0(s)\big)
\equiv F_0(\lambda)\big(\langle Q\rangle^{s-\mu}\otimes1\big)
\big\{(P^2-i)\otimes1+1\otimes(\Delta_\Sigma)_\lambda\big\}R(s,\mu-s)^*
\in\B\big(\H,\H_0(\infty)\big)
$$
depends on $\lambda$ only via the factor $F_0(\lambda)$, since
$(\Delta_\Sigma)_\lambda$ is independent of $\lambda$ on $(\tau_j,\tau_{j+1})$.
Therefore, it follows by Lemma \ref{Lem_F_0}.(b) that the function
$
[0,\infty)\setminus\T\ni\lambda\mapsto Z_0\big(\lambda;G_0(s)\big)
\in\B\big(\H,\H_0(\infty)\big)
$
is locally $k$-times H\"older continuously differentiable if $s$ is chosen such that
$\mu-s>k+1/2$. But, we know by hypothesis that $\mu>k+1$. So, the condition
$\mu-s>k+1/2$ is verified for any $s\in\big(1/2,\mu-k-1/2)\subset(1/2,\mu-1/2)$.

(c) Let $s_1,s_2\in(1/2,\mu-1/2)$, $\varphi\in\H_0(\lambda)$ and
$\psi\in C^\infty_{\rm c}(M)$. Then, Formula \eqref{Form_L} implies that
$$
\big\langle\varphi,
Z_0\big(\lambda;G_0(s_1)\big)\langle\Phi\rangle^{-s_1}\psi\big\rangle_{\H_0(\lambda)}
=\big\langle L(s_1,\mu-s_1)\langle\Phi_0\rangle^{s_1-\mu}F_0(\lambda)^*\varphi,
(H+i)\langle\Phi\rangle^{-s_1}\psi\big\rangle_\H.
$$
So, by taking $\{\zeta_n\}_{n\in\N}\subset\S(\R)\odot C^\infty(\Sigma)$ such that
$
\lim_{n\to\infty}
\|\langle\Phi_0\rangle^{s_1-\mu}F_0(\lambda)^*\varphi-\zeta_n\|_{\H_0}=0
$,
one infers that
\begin{align*}
\big\langle\varphi,
Z_0\big(\lambda;G_0(s_1)\big)\langle\Phi\rangle^{-s_1}\psi\big\rangle_{\H_0(\lambda)}
&=\lim_{n\to\infty}\big\langle(H-i)^{-1}\langle\Phi\rangle^{s_1}T
\langle\Phi_0\rangle^{\mu-s_1}\zeta_n,
(H+i)\langle\Phi\rangle^{-s_1}\psi\big\rangle_\H\\
&=\lim_{n\to\infty}\big\langle(H-i)^{-1}\langle\Phi\rangle^{s_2}T
\langle\Phi_0\rangle^{\mu-s_2}\langle\Phi_0\rangle^{s_2-s_1}\zeta_n,
(H+i)\langle\Phi\rangle^{-s_2}\psi\big\rangle_\H\\
&=\lim_{n\to\infty}\big\langle L(s_2,\mu-s_2)\langle\Phi_0\rangle^{s_2-s_1}\zeta_n,
(H+i)\langle\Phi\rangle^{-s_2}\psi\big\rangle_\H\\
&=\big\langle\varphi,Z_0\big(\lambda;G_0(s_2)\big)\langle\Phi\rangle^{-s_2}\psi
\big\rangle_{\H_0(\lambda)}.
\end{align*}
One concludes by noting that $\varphi$ is arbitrary in $\H_0(\lambda)$ and that
$C^\infty_{\rm c}(M)$ is dense in $\H$.
\end{proof}

In the proof of the next theorem, we use the fact that the identification operator
$J$ extends, for each $s\in\R$, to an element of
$\B\big(\dom(\langle\Phi_0\rangle^s),\dom(\langle\Phi\rangle^s)\big)$. We also use
the notation $\widehat\sigma(H_0)$ for a core of the spectrum
$\sigma(H_0)\equiv\sigma_{\rm ac}(H_0)$; namely, a Borel set $\widehat\sigma(H_0)$
such that:
\begin{enumerate}
\item[(i)] $\widehat\sigma(H_0)$ is a Borel support of the spectral measure
$E^{H_0}(\;\!\cdot\;\!)$, \ie $E^{H_0}\big(\R\setminus\widehat\sigma(H_0)\big)=0$,
\item[(ii)] if $I$ is a Borel support of $E^{H_0}(\;\!\cdot\;\!)$, then
$\widehat\sigma(H_0)\setminus I$ has Lebesgue measure zero.
\end{enumerate}
The set $\widehat\sigma(H_0)$ is unique up to a set of Lebesgue measure zero (see
\cite[Sec.~1.3.3]{Yaf92} for more details).

\begin{Theorem}[Stationary formula for the $S$-matrix]\label{formule_S}
Suppose that Assumption \ref{CondsurgV} holds with $\mu>1$. Then, for any
$s_1,s_2,s_3\in(1/2,\mu-1/2)$ and for almost every
$\lambda\in[0,\infty)\setminus\kappa(H)$ we have
\begin{align}\label{formulette}
S(\lambda)
&=-2\pi iF_0(\lambda)J^*\langle\Phi\rangle^{-s_1}Z_0\big(\lambda;G_0(s_1)\big)^*\\
&\qquad+2\pi iZ_0\big(\lambda;G_0(s_2)\big)\langle\Phi\rangle^{-s_2}
(H-\lambda-i0)^{-1}\langle\Phi\rangle^{-s_3}Z_0\big(\lambda;G_0(s_3)\big)^*,\nonumber
\end{align}
with $Z_0\big(\lambda;G_0(\;\!\cdot\;\!)\big)$ given by the r.h.s. of
\eqref{form_Z_0}.
\end{Theorem}

Before the proof, we recall that the usual scattering operator
$\widetilde S:\H_0\to\H_0$ coincides on $\H_0^-$ with our unitary scattering operator
$S:\H_0^-\to\H_0^+$.

\begin{proof}
Let $s_1\in(1/2,\mu-1/2)$, $\varphi\in\H_0^-\cap\dom(\langle\Phi_0\rangle^{s_1})$ and
$\lambda\in\R\setminus\T$. Then, we know from Lemma \ref{passage} that the following
limits exist in $\H$ (see the proof of Proposition \ref{existence_le_retour} for the
definitions of $G_0(s_1)$ and $C_{s_1}$):
\begin{align*}
&\slim_{\varepsilon\searrow0}G_0(s_1)(H_0-\lambda\mp i\varepsilon)^{-1}\varphi\\
&=\slim_{\varepsilon\searrow0}C_{s_1}\langle\Phi_0\rangle^{s_1-\mu}(H_0-i)
(H_0-\lambda\mp i\varepsilon)^{-1}\langle\Phi_0\rangle^{-s_1}
\langle\Phi_0\rangle^{s_1}\varphi\\
&=\slim_{\varepsilon\searrow0}C_{s_1}\big\{\langle\Phi_0\rangle^{-\mu}
+(\lambda\pm i\varepsilon-i)\langle\Phi_0\rangle^{s_1-\mu}
(H_0-\lambda\mp i\varepsilon)^{-1}\langle\Phi_0\rangle^{-s_1}\big\}
\langle\Phi_0\rangle^{s_1}\varphi.
\end{align*}
Furthermore, the operator $G_0(s_1)$ is $H_0$-smooth in the weak sense since it is
$H_0$-smooth on $\R\setminus\T$ (see Section 5.1 of \cite{Yaf92}), and the operator
$G(s_1)\equiv\langle\Phi\rangle^{-s_1}$ is $|H|^{1/2}$-bounded. Therefore, all the
assumptions of \cite[Thm.~5.5.3]{Yaf92} are verified on the dense set
$\dom(\langle\Phi_0\rangle^{s_1})\subset\H_0$ due to Proposition \ref{yep}. It follows
that the representation \cite[Eq.~(5.5.3$_+$)]{Yaf92} for $\widetilde S(\lambda)$
holds for almost every $\lambda\in\widehat\sigma(H_0)$. So, we have for almost every
$\lambda\in\widehat\sigma(H_0)\setminus\T$ and all
$\varphi\in\H_0^-\cap\dom(\langle\Phi_0\rangle^{s_1})$ the equalities
\begin{align}
\big(F_0\widetilde S\varphi\big)(\lambda)
=\widetilde S(\lambda)F_0(\lambda)\varphi
&=\big\{u_+(\lambda)-2\pi i\big[Z_0\big(\lambda;\widetilde G_0(s_1)\big)
Z_0\big(\lambda;G_0(s_1)\big)^*\label{S_1}\\
&\qquad-Z_0\big(\lambda;G_0(s_1)\big)B_{s_1}(\lambda+i0)
Z_0\big(\lambda;G_0(s_1)\big)^*\big]\big\}F_0(\lambda)\varphi,\nonumber
\end{align}
with the operators $u_+(\lambda)$, $Z_0\big(\lambda;\widetilde G_0(s_1)\big)$ and
$B_{s_1}(\lambda+i0)$ defined in points (i), (ii) and (iii) that follow:

(i) We know from \cite[Thm.~5.3.6]{Yaf92} (which applies in our case) that the
stationary wave operator $\U_+(H,H_0;J)$ coincides with the wave operator $W_+$. It
then follows from \cite[Eq.~(2.7.16)]{Yaf92} that
$$
\U_+(H_0,H_0;J^*J)=\U_+(H,H_0;J)^*\;\!\U_+(H,H_0;J)=W_+^*W_+=P_0^+,
$$
with $P_0^+$ the orthogonal projection onto $\H_0^+$. Since $\H_0^+$ and $\H_0^-$ are
orthogonal and since $u_+(\lambda):\H_0(\lambda)\to\H_0(\lambda)$ is defined by the
relation
$$
u_+(\lambda)F_0(\lambda)\varphi=\big[F_0\;\!\U_+(H_0,H_0;J^*J)\varphi\big](\lambda),
$$
it follows that
$$
u_+(\lambda)F_0(\lambda)\varphi=\big(F_0P_0^+\varphi\big)(\lambda)=0.
$$

(ii) One has $\widetilde G_0(s_1):=G(s_1)J$ with $G(s_1)=\langle\Phi\rangle^{-s_1}$.
Therefore, the operator $Z_0\big(\lambda;\widetilde G_0(s_1)\big):\H\to\H_0(\lambda)$
(defined as $Z_0\big(\lambda;G_0(s_1)\big)$, but with $G_0(s_1)$ replaced by
$\widetilde G_0(s_1)$) satisfies for all $\psi\in\H$
$$
Z_0\big(\lambda;\widetilde G_0(s_1)\big)\psi
=F_0(\lambda)\big\{\widetilde G_0(s_1)\big\}^*\psi
=F_0(\lambda)J^*\langle\Phi\rangle^{-s_1}\psi.
$$
Lemma \ref{Lem_F_0}.(a) and the inclusion
$J^*\in\B\big(\dom(\langle\Phi\rangle^{s_1}),\dom(\langle\Phi_0\rangle^{s_1})\big)$
implies that
$Z_0\big(\lambda;\widetilde G_0(s_1)\big)\in\B\big(\H,\H_0(\lambda)\big)$.

(iii) The operator
$$
B_{s_1}(\lambda+i0)
:=G(s_1)(H-\lambda-i0)^{-1}G(s_1)^*
=\langle\Phi\rangle^{-s_1}(H-\lambda-i0)^{-1}\langle\Phi\rangle^{-s_1}
$$
belongs to $\B(\H)$ for all $\lambda\in\R\setminus\kappa(H)$ due to Proposition
\ref{yep}.

Now, by replacing the expressions of points (i), (ii) and (iii) into \eqref{S_1} and
then by using Lemma \ref{lemme_Z_0}.(c), one gets for any $s_1,s_2,s_3\in(1/2,\mu-1/2)$
and for almost every $\lambda\in\widehat\sigma(H_0)\setminus\kappa(H)$ that
\begin{align*}
\widetilde S(\lambda)F_0(\lambda)\varphi
&=-2\pi i\big\{F_0(\lambda)J^*\langle\Phi\rangle^{-s_1}
Z_0\big(\lambda;G_0(s_1)\big)^*\\
&\hspace{45pt}-Z_0\big(\lambda;G_0(s_1)\big)\langle\Phi\rangle^{-s_1}
(H-\lambda-i0)^{-1}\langle\Phi\rangle^{-s_1}Z_0\big(\lambda;G_0(s_1)\big)^*\big\}
F_0(\lambda)\varphi.\\
&=-2\pi i\big\{F_0(\lambda)J^*\langle\Phi\rangle^{-s_1}
Z_0\big(\lambda;G_0(s_1)\big)^*\\
&\hspace{45pt}-\lim_{\varepsilon\searrow0}Z_0\big(\lambda;G_0(s_1)\big)\langle\Phi\rangle^{-s_1}
(H-\lambda-i\varepsilon)^{-1}\langle\Phi\rangle^{-s_1}Z_0\big(\lambda;G_0(s_1)\big)^*\big\}
F_0(\lambda)\varphi.\\
&=-2\pi i\big\{F_0(\lambda)J^*\langle\Phi\rangle^{-s_1}
Z_0\big(\lambda;G_0(s_1)\big)^*\\
&\hspace{45pt}-Z_0\big(\lambda;G_0(s_2)\big)\langle\Phi\rangle^{-s_2}
(H-\lambda-i0)^{-1}\langle\Phi\rangle^{-s_3}Z_0\big(\lambda;G_0(s_3)\big)^*\big\}
F_0(\lambda)\varphi.
\end{align*}
Furthermore, Lemma \ref{Lem_F_0}.(a), Lemma \ref{lemme_Z_0} and Proposition \ref{yep}
imply that the operator within the curly brackets is well-defined on
$\H_0^-(\lambda)$ for all $\lambda\in[0,\infty)\setminus\kappa(H)$. So, since
$\widetilde S$ and $S$ are equal on $\H_0^-$, it follows that  \eqref{formulette}
holds for almost every $\lambda\in[0,\infty)\setminus\kappa(H)$.
\end{proof}

In the next corollary, we identify (without loss of generality) the operator
$S(\lambda)$ with the r.h.s. of Formula \eqref{formulette} for all
$\lambda\in[0,\infty)\setminus\kappa(H)$.

\begin{Corollary}[Differentiability of the $S$-matrix]\label{Cor_diff}
Suppose that Assumption \ref{CondsurgV} holds with $\mu>k+1$ for some $k\in\N$.
Then, the function
$
[0,\infty)\setminus\kappa(H)\ni\lambda\mapsto S(\lambda)
\in\B\big(\H_0^-(\infty),\H_0^+(\infty)\big)
$
is locally $k$-times H\"older continuously differentiable.
\end{Corollary}

\begin{proof}
We first show that $\lambda\mapsto S(\lambda)$ is locally $k$-times H\"older
continuously differentiable from $[0,\infty)\setminus\kappa(H)$ to
$\B\big(\H_0^-(\infty),\H_0(\infty)\big)$. For that purpose, we let
$s_1,s_2,s_3\in(1/2,\mu-1/2)$ and note from Formula \eqref{formulette} that it is
sufficient to prove that the terms
$$
A_{\ell_1,\ell_2}(\lambda)
:=\left\{\frac{\d^{\ell_1}}{\d\lambda^{\ell_1}}\;\!F_0(\lambda)J^*
\langle\Phi\rangle^{-s_1}\right\}\left\{\frac{\d^{\ell_2}}
{\d\lambda^{\ell_2}}\;\!Z_0\big(\lambda;G_0(s_1)\big)^*\right\}
$$
exist and are locally H\"older continuous for all
$\lambda\in[0,\infty)\setminus\kappa(H)$ and all integers $\ell_1,\ell_2\ge0$
satisfying $\ell_1+\ell_2\le k$, and that the terms
$$
B_{\ell_1,\ell_2,\ell_3}(\lambda)
:=\left\{\frac{\d^{\ell_1}}{\d\lambda^{\ell_1}}\;\!
Z_0\big(\lambda;G_0(s_2)\big)\right\}\left\{\frac{\d^{\ell_2}}{\d\lambda^{\ell_2}}
\;\!\langle\Phi\rangle^{-s_2}(H-\lambda-i0)^{-1}\langle\Phi\rangle^{-s_3}\right\}
\left\{\frac{\d^{\ell_3}}{\d\lambda^{\ell_3}}\;\!
Z_0\big(\lambda;G_0(s_3)\big)^*\right\}
$$
exist and are locally H\"older continuous for all
$\lambda\in[0,\infty)\setminus\kappa(H)$ and all integers $\ell_1,\ell_2,\ell_3\ge0$
satisfying $\ell_1+\ell_2+\ell_3\le k$.

Now, the factors in $B_{\ell_1,\ell_2,\ell_3}(\lambda)$ satisfy
\begin{align*}
\frac{\d^{\ell_3}}{\d\lambda^{\ell_2}}\;\!Z_0\big(\lambda;G_0(s_3)\big)^*
\in\B\big(\H_0^-(\infty),\H\big) \quad&\hbox{for }s_3 \in (1/2,\mu-\ell_3-1/2),\\
\frac{\d^{\ell_2}}{\d\lambda^{\ell_2}}\;\!\langle\Phi\rangle^{-s_2}
(H-\lambda-i0)^{-1}\langle\Phi\rangle^{-s_3}\in\B(\H)
\quad&\hbox{for }s_2,s_3>\ell_2+1/2,\\
\frac{\d^{\ell_1}}{\d\lambda^{\ell_1}}\;\!Z_0\big(\lambda;G_0(s_2)\big)
\in\B\big(\H,\H_0(\infty)\big) \quad&\hbox{for }s_2 \in (1/2,\mu-\ell_1-1/2),
\end{align*}
and are locally H\"older continuous due to Proposition \ref{yep} and Lemma
\ref{lemme_Z_0}. Therefore, if
$$
s_2,s_3\in(\ell_2+1/2,\ell_2+1/2 +\mu-k-1)\subset(1/2,\mu-1/2),
$$
then $B_{\ell_1,\ell_2,\ell_3}(\lambda)$ exists and is locally H\"older continuous
for all $\lambda\in[0,\infty)\setminus\kappa(H)$. Since similar arguments apply to
the term $A_{\ell_1,\ell_2}(\lambda)$ if $s_1\in(\ell_1+1/2,\ell_1+1/2 +\mu-k-1)$,
the announced differentiability is proved.

To conclude the proof, it only remains to note that all the derivatives
$\frac{\d^\ell}{\d\lambda^\ell}S(\lambda)$, $\ell\in\{1,\ldots,k\}$, map
$\H_0^-(\lambda)$ into $\H_0^+(\lambda)$ due to the formula
$$
S(\lambda)\;\!\H_0^-(\lambda)
=\big(F_0P_0^+S\;\!\H_0^-\big)(\lambda)
=P_0^+(\lambda)S(\lambda)\;\!\H_0^-(\lambda)
$$
with $P_0^+(\lambda):=\big(F_0P_0^+F_0^{-1}\big)(\lambda)$.
\end{proof}

\subsection{Mapping properties of the scattering operator}

In this subsection, we define and give some properties of a subset $\E\subset\H_0^-$
which will be useful when proving the existence of time delay.

Let $\varphi\in\H_0^-$ satisfy $F_0(\lambda)\varphi=\rho(\lambda)h(\lambda)$ for each
$\lambda\in[0,\infty)\setminus\T$, where $\rho\in C^\infty\big([0,\infty)\big)$ has
compact support in $[0,\infty)\setminus\kappa(H)$ and
$[0,\infty)\setminus\kappa(H)\ni\lambda\mapsto h(\lambda)\in\H_0(\lambda)$ is
$\lambda$-independent on each interval of $[0,\infty)\setminus \kappa(H)$. Then, the
finite span $\E$ of such vectors is dense in $\H_0^-$ if Assumption \ref{CondsurgV}
holds with $\mu>1$ (see Proposition \ref{natureH}), and we have the following
inclusions:

\begin{Proposition}\label{Prop_map}
Suppose that Assumption \ref{CondsurgV} holds with $\mu>4$. Then $\E\subset\D_3$ and
$S\;\!\E\subset\D_3$.
\end{Proposition}

\begin{proof}
If $\varphi\in\E$, there exists a compact set $I$ in $[0,\infty)\setminus\kappa(H)$
such that $E^{H_0}(I)\varphi=\varphi$. Thus, in order to show that $\varphi\in\D_3$
one has to verify that $\varphi\in\H_3(\R)\otimes\ltwo(\Sigma)=\dom(Q^3\otimes1)$.
So, let $\psi\in\S(\R)\odot\ltwo(\Sigma)$. Then, using \eqref{eq_T_0} and Lemma
\ref{Lem_F_0}.(c), we obtain for each $\lambda\in[0,\infty)\setminus\T$
\begin{equation}\label{Q3}
\big[F_0\big(Q^3\otimes1\big)\psi\big](\lambda)_j
=\big\{i\zeta(\lambda)_j^-,-i\zeta(\lambda)_j^+\big\},
\end{equation}
where
\begin{align}
\zeta^\pm_j(\lambda)
&:=\frac38(\lambda-\tau_j)^{-3/2}(F_0\psi)(\lambda)^\pm_j
+\frac32(\lambda-\tau_j)^{-1/2}
\frac\d{\d\lambda}(F_0\psi)(\lambda)^\pm_j\nonumber\\
&\qquad+18(\lambda-\tau_j)^{1/2}
\frac{\d^2}{\d\lambda^2}(F_0\psi)(\lambda)^\pm_j
+8(\lambda-\tau_j)^{3/2}
\frac{\d^3}{\d\lambda^3}(F_0\psi)(\lambda)^\pm_j.\label{Q3bis}
\end{align}
The r.h.s. of \eqref{Q3}-\eqref{Q3bis} with $\psi\in\S(\R)\odot\ltwo(\Sigma)$
replaced by $\varphi\in\E$ defines a vector $\widetilde\varphi$ belonging to
$\widehat\H_0$. Thus, using partial integration for the terms involving derivatives
with respect to $\lambda$, one finds that
$$
\left|\big\langle(Q^3\otimes1)\psi,\varphi\big\rangle_{\H_0}\right|
=\left|\big\langle F_0\psi,\widetilde\varphi\big\rangle_{\widehat \H_0}\right|
\le{\rm Const.}\;\!\|\psi\|_{\H_0}
$$
for all $\psi\in\S(\R)\odot\ltwo(\Sigma)$. Since $Q^3\otimes1$ is essentially
self-adjoint on $\S(\R)\odot\ltwo(\Sigma)$, this implies that
$\varphi\in\dom(Q^3\otimes1)$, and therefore the inclusion $\E\subset\D_3$.

For the second inclusion $S\;\!\E\subset\D_3$, one observes that the function
$
[0,\infty)\setminus\kappa(H)\ni\lambda\mapsto S(\lambda)
\in\B\big(\H_0^-(\infty),\H_0^+(\infty)\big)
$
is locally $3$-times H\"older continuously differentiable due to Corollary
\ref{Cor_diff}. Thus, the above argument with $\varphi$ replaced by $S\varphi$ gives
the result.
\end{proof}

\begin{Remark}
We believe that the statement of Proposition \ref{Prop_map} could be replaced by the
following more general statement but we could not find a simple proof for it: Suppose
that Assumption \ref{CondsurgV} holds with $\mu>3$, then there exists $s>2$ such that
$\E\subset\D_s$ and $S\;\!\E\subset\D_s$. Such a result would lead to better mapping
properties of the scattering operator, and thus the necessary assumption on $\mu$ for
the existence of the time delay in the next section could be weakened accordingly.
\end{Remark}

\subsection{Time delay}\label{Sec_Time}

We introduce in this section the notion time delay defined in terms of sojourn times,
and then we prove its existence and its equality with the so-called Eisenbud-Wigner
time delay. All proofs are based on the abstract framework developed in \cite{RT11}
and on the various estimates obtained in the previous sections.

We define the sojourn times by particularising to our present model the definitions
of \cite{RT11}. For that purpose, we start by choosing a position observable in
$\H_0$ which satisfies the special relations with respect to $H_0$ required in
\cite[Sec.~2]{RT11}. The most natural choice is the position operator
$\Phi_0\equiv Q\otimes1$ along the $\R$-axis of $\R\times\Sigma$ already introduced
in \eqref{defphi0}. Then, we define the sojourn time for the free evolution
$\e^{-itH_0}$ as follows: Given $\chi_{[-1,1]}$ the characteristic function for the
set $[-1,1]$, we set for $\varphi\in \D_0$ and $r>0$
$$
T_r^0(\varphi):=\int_\R\d t\,\big\langle\e^{-itH_0}\varphi,
\chi_{[-1,1]}(\Phi_0/r)\e^{-itH_0}\varphi\big\rangle_{\H_0},
$$
where the integral has to be understood as an improper Riemann integral. The operator
$\chi_{[-1,1]}(\Phi_0/r)$ is the projection onto the subspace
$E^{\Phi_0}([-r,r])\H_0$ of states localised on the cylinder $[-r,r]\times\Sigma$.
Therefore, if $\|\varphi\|_{\H_0}=1$, then $T_r^0(\varphi)$ can be interpreted as the
time spent by the evolving state $\e^{-itH_0}\varphi$ inside $[-r,r]\times\Sigma$.

When defining the sojourn time for the full evolution $\e^{-itH}$, one faces the
problem that the localisation operator $\chi_{[-1,1]}(\Phi_0/r)$ acts in $\H_0$ while
the operator $\e^{-itH}$ acts in $\H$. The simplest solution to this problem is to
consider the operator $\chi_{[-1,1]}(\Phi_0/r)$ injected in $\H$ via $J$, \ie
$J\chi_{[-1,1]}(\Phi_0/r)J^*\in\B(\H)$, and for the present model this solution turns
out to be appropriate (see nonetheless \cite[Sec.~4]{RT11} for a more general
approach). It is then natural to define the sojourn time for the full evolution
$\e^{-itH}$ by the expression
\begin{equation*}
T_{r,1}(\varphi):=\int_\R\d t\,\big\langle \e^{-itH}W_-\varphi,
J\chi_{[-1,1]}(\Phi_0/r)J^*\e^{-itH}W_-\varphi\big\rangle_\H.
\end{equation*}
Another sojourn time appearing in this context is
\begin{equation*}
T_2(\varphi):=\int_\R\d t\,\big\langle\e^{-itH}W_-\varphi,\big(1-JJ^*\big)
\e^{-itH}W_-\varphi\big\rangle_\H.
\end{equation*}
The finiteness of these expressions is proved below for suitable $\varphi$ under
Assumption \ref{CondsurgV} with $\mu$ big enough. The term $T_{r,1}(\varphi)$ can be
interpreted as the time spent by the scattering state $\e^{-itH}W_-\varphi$, injected
in $\H_0$ via $J^*$, inside $E^{\Phi_0}([-r,r])\H_0$. The term $T_2(\varphi)$ can be
seen as the time spent by the scattering state $\e^{-itH}W_-\varphi$ inside the
subset $\big(1-JJ^*\big)\H$ of $\H$. Roughly speaking, this corresponds to the time
spent by the state in the relatively compact set $\Mc\subset M$.
Within this framework, we say that
\begin{equation*}
\tau_r(\varphi):=T_r(\varphi)-\12\big\{T_r^0(\varphi)+T_r^0(S\varphi)\big\},
\end{equation*}
with $T_r(\varphi):=T_{r,1}(\varphi)+T_2(\varphi)$, is the symmetrized time delay of
the scattering system $(H_0,H,J)$ with incoming state $\varphi$. This symmetrized
version of the usual time delay
$$
\tau_r^{\rm in}(\varphi):=T_r(\varphi)-T_r^0(\varphi)
$$
is known to be the only time delay having a well-defined limit as $r\to\infty$ for
complicated scattering systems (see for example
\cite{AJ07,BO79,GT07,Mar75,Mar81,SM92,Smi60,Tie06}). Our main result, properly stated
below, is thus the existence of the limit $\lim_{r\to\infty}\tau_r(\varphi)$ and its
identity with the Eisenbud-Wigner time delay which we now define.

Given a localisation function $f:\R\to[0,\infty)$ and an abstract pair of operators
$(H_0,\Phi_0)$ satisfying some compatibility assumptions, it is shown in \cite{RT10}
how to construct a natural time operator $T_f$ for $H_0$. Now, for the localisation
function $f=\chi_{[-1,1]}$ and for our pair $(H_0,\Phi_0)$ of operators, this
abstract construction simplifies drastically, and a rapid inspection of
\cite[Prop.~2.6.(b)]{Tie06} and \cite[Thm.~5.5]{RT10} shows that the general time
operator $T_f$ introduced in \cite[Sec.~5]{RT10} reduces to the operator $T$ given by
\begin{equation}\label{nemenveutpas}
\langle \varphi,T\varphi\rangle_{\H_0}:=
\textstyle\big\langle\varphi,
\frac14\big(QP^{-1}+P^{-1}Q\big)\otimes1
\varphi\big\rangle_{\H_0},\quad\varphi\in \D_1.
\end{equation}
The operator $\frac14\big(QP^{-1}+P^{-1}Q\big)$,
known as the Aharonov-Bohm operator, is the usual time operator for the
one-dimensional Schr\"odinger operator $P^2$ (see \cite[Sec.~1]{AB61} and
\cite[Sec.~1]{Miy01}).

We are now in a suitable position to prove the existence of the limit
$\lim_{r\to\infty}\tau_r(\varphi)$ for incoming states $\varphi$ in the dense subset
$\E\subset\H_0^-$ introduced in the previous section:

\begin{Theorem}[Existence of time delay]\label{thm_main}
Suppose that Assumption \ref{CondsurgV} holds with $\mu> 4$. Then, one has for each
$\varphi\in\E$
\begin{equation}\label{Eisenbud_sym}
\lim_{r\to\infty}\tau_r(\varphi)
=-\big\langle\varphi,S^*\big[T,S\big]\varphi\big\rangle_{\H_0},
\end{equation}
with $T$ given by \eqref{nemenveutpas}.
\end{Theorem}

\begin{proof}
The proof consists in an application of the abstract result \cite[Thm.~4.3]{RT11}.
However, we first have to note that this theorem also applies to our non-smooth
localisation function $f=\chi_{[-1,1]}$. Indeed, the only points where the smoothness
of the localisation $f$ is required in the proof of \cite[Thm.~4.3]{RT11} is for
applying Theorem 3.4 and Lemma 4.2 of \cite{RT11}. Now, the result of
\cite[Thm.~3.4]{RT11} also holds for $f=\chi_{[-1,1]}$ due to
\cite[Prop.~2.6.(b)]{Tie06}, and a rapid inspection of \cite[Lemma 4.2]{RT11} shows
that its proof also holds for $f=\chi_{[-1,1]}$. So, Theorem 4.3 of \cite{RT11} can
be applied, and we are left with the verification of its assumptions.

For that purpose, one first observes that with our choice of operator $\Phi_0$, one
has for each $x\in\R$
$$
H_0(x):=\e^{-ix\Phi_0}H_0\e^{ix\Phi_0}
=(P+x)^2\otimes1+1\otimes\Delta_\Sigma.
$$
Therefore, the operators $H_0(x)$, $x\in\R$, mutually commute (Assumption 2.1 of
\cite[Thm.~4.3]{RT11}), and the regularity of $H_0$ with respect to $\Phi_0$ is
easily checked (Assumption 2.2 of \cite[Thm.~4.3]{RT11}). In addition, a direct
calculation using \eqref{decompo} shows that the set $\kappa(H_0)$ of critical values
of $H_0$, introduced in \cite[Def.~2.3]{RT11}, coincides with $\T$. Furthermore, it
follows from Proposition \ref{Prop_map} that $\varphi\in\H_0^-\cap\D_3$ and
$S\varphi\in\D_3$. Finally, since $S\varphi$ also belongs to $\H_0^+$, it follows
from Lemma \ref{condL1} that both conditions of \cite[Eq.~(4.6)]{RT11} are satisfied.
Thus, Theorem 4.3 of \cite{RT11} applies and leads to the claim.
\end{proof}

The interest of the equality between both definitions of time delay  is twofold. It
generalizes and unifies various results on time delay scattered in the literature.
And it establishes a relation between the two formulations of scattering theory:
Eisenbud-Wigner time delay is a product of the stationary formulation while
expressions involving sojourn times are defined using the time-dependent formulation.
An equality relating these two formulations is always welcome.

\begin{Remark}
Since $T$ is equal to the Aharonov-Bohm operator \eqref{nemenveutpas}, the r.h.s. of
\eqref{Eisenbud_sym} can be even further simplified. Indeed, following
\cite[Rem.~2.7]{Tie06} one can check that the operator $F_0TF_0^{-1}$ acts as
$i\frac\d{\d\lambda}$ in the spectral representation of $H_0$. Thus, under the
hypotheses of Theorem \ref{thm_main}, the relation \eqref{Eisenbud_sym} reads
\begin{equation*}
\lim_{r\to\infty}\tau_r(\varphi)
=\int_0^\infty \d \lambda\,\left\langle(F_0\varphi)(\lambda),
-iS(\lambda)^*\left(\frac{\d S(\lambda)}{\d \lambda}\right)(F_0\varphi)(\lambda)
\right\rangle_{\H_0(\lambda)}.
\end{equation*}
\end{Remark}

\begin{Remark}
We emphasize that the symmetrized time delay is the only global time delay existing
in our framework. Indeed, as in the case of quantum waveguides \cite{Tie06}, the
scattering process does preserve the total energy $H_0$ but does not preserve the
longitudinal kinetic energy $P^2\otimes1$ alone (rearrangements between the
transverse and longitudinal components of the energy occur during the scattering).
This is in agreement with the general criterion \cite[Thm.~5.3]{RT11} which, here,
implies that the unsymmetrized time delay with incoming state $\varphi\in\E$ exists
if $\big[P^2\otimes1,S\big]\varphi=0$.
\end{Remark}

\section{Appendix}\label{Sec_Appendix}
\setcounter{equation}{0}

We prove in this section various mapping properties of the operators $H_0$ and $H$.
We start with a rather elementary lemma on the position operator $Q$ and the momentum
operator $P$ in $\ltwo(\R)$.

\begin{Lemma}\label{lemme_borne}
Take $s,\tau\ge0$ and $z\in\C\setminus[0,\infty)$. Then, there exists a constant
$\textsc c\equiv\textsc c(s,z)>0$ independent of $\tau$ such that
$$
\big\|(P^2+\tau-z)\langle Q\rangle^{-s}(P^2+\tau-z)^{-1}
\langle Q\rangle^s\big\|_{\B(\ltwo(\R))}\le\textsc c.
$$
\end{Lemma}

\begin{proof}
First, one observes that
$
(P^2+\tau-z)\langle Q\rangle^{-s}(P^2+\tau-z)^{-1}\langle Q\rangle^s
$
belongs to $\B\big(\ltwo(\R)\big)$ due standard properties of the weighted Sobolev
spaces defined in terms of $\langle Q\rangle$ and $\langle P\rangle$ (see
\cite[Sec.~4.1]{ABG}). Furthermore, one has on $\S(\R)$ the equalities
\begin{align*}
(P^2+\tau-z)\langle Q\rangle^{-s}(P^2+\tau-z)^{-1}
\langle Q\rangle^s
&=1+(P^2+\tau-z)\big[\langle Q\rangle^{-s},(P^2+\tau-z)^{-1}\big]
\langle Q\rangle^s\\
&=1+\big[P^2,\langle Q\rangle^{-s}\big]\langle P\rangle^{-1}
\langle Q\rangle^s\langle Q\rangle^{-s}\langle P\rangle
(P^2+\tau-z)^{-1}\langle Q\rangle^s\\
&=1+B\langle Q\rangle^{-s}\langle P\rangle(P^2+\tau-z)^{-1}
\langle Q\rangle^s,
\end{align*}
with $B:=\big[P^2,\langle Q\rangle^{-s}\big]\langle P\rangle^{-1}\langle Q\rangle^s$
bounded and independent of $\tau$. Therefore, in order to prove the claim it is
sufficient to show that the bounded operator
$
\langle Q\rangle^{-s}\langle P\rangle(P^2+\tau-z)^{-1}\langle Q\rangle^s
$
has its norm dominated by a constant independent of $\tau$. This can easily be done
either by induction on $s$ or by computing iteratively the commutator of
$(P^2+\tau-z)^{-1}$ with $\langle Q\rangle^s$. Details are left to the reader.
\end{proof}

For the next proposition, we recall that $H_0$ and $\Phi_0$ satisfy
$H_0=P^2\otimes1+1\otimes\Delta_\Sigma$ and $\Phi_0=Q\otimes1$ in $\H_0$.

\begin{Proposition}\label{mapping_H_not}
Let $z\in\C\setminus[0,\infty)$, then
\begin{enumerate}
\item[(i)] for any $s\ge0$ the operator
$
\langle H_0\rangle\langle\Phi_0\rangle^{-s}(H_0-z)^{-1}
\langle\Phi_0\rangle^s
$,
defined on $\dom\big(\langle\Phi_0\rangle^s\big)$, is well-defined and extends
continuously to an element of $\B(\H_0)$,
\item[(ii)] $(H_0-z)^{-1}$ belongs to
$\B\big(\dom(\langle\Phi_0\rangle^t),\dom(\langle\Phi_0\rangle^t)\big)$ for each
$t\in\R$,
\item[(iii)] one has the inclusion
$
(H_0-z)^{-1}\big(\S(\R)\odot C^\infty(\Sigma)\big)
\subset\big(\S(\R)\odot C^\infty(\Sigma)\big)
$.
\end{enumerate}
\end{Proposition}

\begin{proof}
(i) Let $\tau_j\in\T$. Then one has
\begin{align*}
&\big\|\langle P^2+\tau_j\rangle\langle Q\rangle^{-s}(P^2+\tau_j-z)^{-1}
\langle Q\rangle^s\big\|_{\B(\ltwo(\R))}\\
&\le\big\|\langle P^2+\tau_j\rangle(P^2+\tau_j-z)^{-1}\big\|_{\B(\ltwo(\R))}
\cdot\big\|(P^2+\tau_j-z)\langle Q\rangle^{-s}
(P^2+\tau_j-z)^{-1}\langle Q\rangle^s\big\|_{\B(\ltwo(\R))}\\
&\le\textsc c.
\end{align*}
for some constant $\textsc c>0$ independent of $\tau_j$, due to Lemma
\ref{lemme_borne}. Therefore, for each $N\in\N$ the operator
$$
F_N:=\sum_{j\le N}\langle P^2+\tau_j\rangle\langle Q\rangle^{-s}
(P^2+\tau_j-z)^{-1}\langle Q\rangle^s\otimes\P_j,
$$
with $\P_j$ the orthogonal projection in $\ltwo(\Sigma)$ associated with $\tau_j$, is
bounded in $\H_0$. Furthermore, a direct calculation using the fact that
$\slim_{N\to\infty}\sum_{j\le N}(1\otimes\P_j)=1$ shows that the norm of $F_N$ is
bounded by a constant independent of $N$ and that the limit $\slim_{N\to\infty}F_N$
exists and is equal to
$
\langle H_0\rangle\langle\Phi_0\rangle^{-s}(H_0-z)^{-1}
\langle\Phi_0\rangle^s$ on $\dom\big(\langle\Phi_0\rangle^s\big)$.
This implies the claim.

(ii) This statement is a direct consequence of \cite[Prop.~5.3.1]{ABG}, which can be
applied since $H_0$ is of class $C^\infty(\Phi_0)$.

(iii) Let $\varphi\in\S(\R)\odot C^\infty(\Sigma)$. Then $(H_0-z)^{-1}\varphi$ is
$C^\infty$ over $\R$ (resp. $\Sigma$) due to the commutation of $(H_0-z)^{-1}$ with
$\langle P\rangle^{-1}\otimes1$ (resp. $1\otimes\langle\Delta_\Sigma\rangle^{-1}$).
The fast decay of $(H_0-z)^{-1}\varphi$ in the $\R$-coordinate follows from point
(ii).
\end{proof}

\begin{Lemma}\label{estimation}
Let $z\in\C\setminus\sigma(H)$, $m\in\N$ and $s\in[0,2m]$. Then, the operator
$\langle A\rangle^{s}(H-z)^{-m}\langle\Phi\rangle^{-s}$ belongs to $\B(\H)$.
\end{Lemma}

\begin{proof}
(i) We start by proving the boundedness of
$\langle A\rangle^{2m}(H-z)^{-m}\langle\Phi\rangle^{-2m}$.

Consider the family of multiplication operators $\chi_n\in\B(\H)$ defined in the
proof of Lemma \ref{regul}. Then $\slim_{n\to\infty}\chi_n=1$, and one has for each
$\varphi\in C^\infty_{\rm c}(M)$ and $n\in\N^*$ that
$\chi_n(H-z)^{-m}\langle\Phi\rangle^{-1}\varphi\in C^\infty_{\rm c}(M)$ due to Lemma
\ref{ellipticite}. Therefore,
\begin{align*}
\big[(H-z)^{-m},\langle\Phi\rangle^{-1}\big]\varphi
&=\langle\Phi\rangle^{-1}(H-z)^{-m}\big[(H-z)^m,\langle\Phi\rangle\big](H-z)^{-m}
\langle\Phi\rangle^{-1}\varphi\\
&=\lim_{n\to\infty}\langle\Phi\rangle^{-1}(H-z)^{-m}
\big[(H-z)^m,\langle\Phi\rangle\big]\chi_n(H-z)^{-m}\langle\Phi\rangle^{-1}\varphi,\\
&=\lim_{n\to\infty}\langle\Phi\rangle^{-1}(H-z)^{-m}L_{2m-1}\;\!\chi_n(H-z)^{-m}
\langle\Phi\rangle^{-1}\varphi,
\end{align*}
with $L_{2m-1}$ a differential operator of order $2m-1$ on $C^\infty_{\rm c}(M)$ with
coefficients in $C_{\rm b}^\infty(M)$. Now, $L_{2m-1}$ extends continuously to a
bounded operator (denoted similarly) from $\H^{2m-1}(M)$ to $\H$ by
\cite[Lemma~1.6]{Sal01}. So, $(H-z)^{-m}L_{2m-1}\in\B(\H)$ and
$L_{2m-1}(H-z)^{-m}\in\B(\H)$, which implies
$$
\big[(H-z)^{-m},\langle\Phi\rangle^{-1}\big]\varphi
=\langle\Phi\rangle^{-1}(H-z)^{-m}L_{2m-1}(H-z)^{-m}\langle\Phi\rangle^{-1}\varphi
$$
and
\begin{align*}
(H-z)^{-m}\langle\Phi\rangle^{-1}\varphi
&=\langle\Phi\rangle^{-1}(H-z)^{-m}
+\big[(H-z)^{-m},\langle\Phi\rangle^{-1}\big]\varphi\\
&=\langle\Phi\rangle^{-1}(H-z)^{-m}
\big\{1+L_{2m-1}(H-z)^{-m}\langle\Phi\rangle^{-1}\big\}\varphi.
\end{align*}
Obviously, one can reproduce those computations to calculate
$(H-z)^{-m}\langle\Phi\rangle^{-k}\varphi$ for any $k=1,2,\ldots,2m$. The result for
$k=2m$ is the following: There exists an operator $B_{2m}\in\B(\H)$ and a sequence
$\big\{B_{2m}^{(n)}\big\}\subset\B(\H)$ with (i)
$B_{2m}^{(n)}C^\infty_{\rm c}(M)\subset C^\infty_{\rm c}(M)$ and (ii)
$\slim_{n\to\infty}B_{2m}^{(n)}=B_{2m}$ on $C^\infty_{\rm c}(M)$ such that
$$
(H-z)^{-m}\langle\Phi\rangle^{-2m}\varphi
=\langle\Phi\rangle^{-2m}(H-z)^{-m}B_{2m}\varphi
$$
for each $\varphi\in C^\infty_{\rm c}(M)$. In particular, one has
$\chi_k(H-z)^{-m}B_{2m}^{(n)}\varphi\in C^\infty_{\rm c}(M)$ for each $k,n\in\N^*$
and $\varphi\in C^\infty_{\rm c}(M)$, and
\begin{align}
(H-z)^{-m}\langle\Phi\rangle^{-2m}\varphi
&=\lim_{k\to\infty}\lim_{n\to\infty}(A+i)^{-2m}(A+i)^{2m}\langle\Phi\rangle^{-2m}
\chi_k(H-z)^{-m}B_{2m}^{(n)}\varphi\nonumber\\
&=\lim_{k\to\infty}\lim_{n\to\infty}
(A+i)^{-2m}L_{2m}\;\!\chi_k(H-z)^{-m}B_{2m}^{(n)}\varphi,\label{pifette}
\end{align}
with $L_{2m}$ a differential operator of order $2m$ on $C^\infty_{\rm c}(M)$ with
coefficients in $C_{\rm b}^\infty(M)$. Now, the extension (denoted similarly) of
$L_{2m}$ to an element of $\B\big(\H^{2m}(M),\H\big)$ satisfies
$(A+i)^{-2m}L_{2m}\in\B(\H)$ and $L_{2m}(H-z)^{-m}\in\B(\H)$. Therefore, one infers
from \eqref{pifette} that
$$
(H-z)^{-m}\langle\Phi\rangle^{-2m}\varphi
=(A+i)^{-2m}L_{2m}(H-z)^{-m}B_{2m}\varphi
=(A+i)^{-2m}B\varphi,
$$
with $B:=L_{2m}(H-z)^{-m}B_{2m}\in\B(\H)$. Since all operators are bounded, this last
equality extends to all $\varphi\in\H$. So, the operator
$\langle A\rangle^{2m}(H-z)^{-m}\langle\Phi\rangle^{-2m}$ can be written as the
product of two bounded operators:
$$
\langle A\rangle^{2m}(H-z)^{-m}\langle\Phi\rangle^{-2m}
\equiv\langle A\rangle^{2m}(A+i)^{-2m}\cdot B.
$$

(ii) Let $R_1:=\langle\Phi\rangle^{-2m}$, $X:=(H-\bar z)^{-m}$ and
$R_2:=\langle A\rangle^{2m}$. Then, point (i) implies that the closure of
$R_1XR_2\upharpoonright\dom(R_2)$ belongs to $\B(\H)$. Since $R_1,R_2$ are positive
invertible self-adjoint operators with $R_1\in\B(\H)$, and $X\in\B(\H)$, one can
apply interpolation (see for example \cite[Prop.~6.17]{Amr09}) to infer that $R_2^\nu
X^*R_1^\nu\in\B(\H)$
for all $\nu\in[0,1]$. However, this implies nothing else but the desired
inclusion;
namely, $\langle A\rangle^{s}(H-z)^{-m}\langle\Phi\rangle^{-s}\in\B(\H)$
for all
$s\in[0,2m]$.
\end{proof}


\def\cprime{$'$}

\end{document}